\documentclass[aip,jmp,11pt,a4paper,byrevtex]{revtex4-1}
\usepackage{amsmath}
\usepackage{amsfonts}
\usepackage{amssymb}
\usepackage{graphicx}

\providecommand{\U}[1]{\protect\rule{.1in}{.1in}}
\newtheorem{theorem}{Theorem}
\newtheorem{acknowledgement}{Acknowledgement}

\newtheorem{corollary}{Corollary}

\newtheorem{lemma}{Lemma}

\newtheorem{proposition}{Proposition}
\newtheorem{remark}{Remark}

\newenvironment{proof}[1][Proof]{\noindent\textbf{#1.} }{\ \rule{0.5em}{0.5em}}
\input{tcilatex}
\begin{document}

\title{Context-invariant quasi hidden variable (qHV) modelling of all joint
von Neumann measurements for an arbitrary Hilbert space}
\author{Elena R. Loubenets}
\affiliation{Moscow State Institute of Electronics and Mathematics, \\
Moscow 109028, Russia}

\begin{abstract}
We prove the existence for each Hilbert space of the two new quasi hidden
variable (qHV) models, \emph{statistically} \emph{noncontextual and
context-invariant, }reproducing all the von Neumann joint probabilities via
nonnegative values of real-valued measures and all the quantum product
expectations -- via the qHV (classical-like) average of the product of the
corresponding random variables. In a context-invariant model, a quantum
observable $X$\ can be represented by a variety of random variables
satisfying the functional condition required in quantum foundations but each
of these random variables equivalently models $X$\ under all joint von
Neumann measurements, regardless of their contexts. The\ proved\ existence
of this model \emph{negates the general opinion} that, in terms of random
variables, the Hilbert space description of all the joint von Neumann
measurements for $\dim \mathcal{H}\geq 3$ can be reproduced only
contextually. The existence of a statistically noncontextual qHV\ model, in
particular, implies that \emph{every N-partite quantum state admits a local
quasi hidden variable (LqHV) model }introduced in [Loubenets, \emph{J. Math.
Phys.} \textbf{53}, 022201 (2012)]. The new results of the present paper
point also to the generality of the quasi-classical probability model
proposed in [Loubenets, \emph{J. Phys. A: Math}. \emph{Theor}. \textbf{45},
185306 (2012)].
\end{abstract}

\maketitle
\tableofcontents

\section{Introduction}

The relation between the quantum probability model and the classical
probability model has been a point of intensive discussions ever since the
seminal publications of von Neumann \cite{1}, Kolmogorov \cite{7}, and
Einstein, Podolsky and Rosen (EPR)\ \cite{2}. In the frame of the quantum
formalism\textbf{, }the interpretation of von Neumann measurements in
classical probability terms, that is, via random variables and probability
measures on a measurable space \cite{4-} $(\Omega ,\mathcal{F}_{\Omega }),$
is generally referred to as \emph{a hidden variable (HV) model}, but a
setting of a HV\ model depends essentially on its aim. Moreover, in the
literature, the HV models are usually divided into noncontextual and
contextual.

\emph{In a noncontextual model}, each quantum observable $X$ is represented
on $(\Omega ,\mathcal{F}_{\Omega })$ \emph{by only one} random variable $%
f_{X}$ with values in the spectrum \textrm{sp}$X$ of this observable $X.$

\emph{In a contextual model}, not only there are quantum observables $%
X_{\gamma },$ $\gamma \in \Upsilon ,$ each modelled by a variety of random
variables on $(\Omega ,\mathcal{F}_{\Omega }),$ but also -- which of these
random variables represents an observable $X_{\gamma }$ under a joint von
Neumann measurement depends specifically on a \emph{context} of this
measurement, i. e. on other compatible quantum observables measured jointly
with $X_{\gamma }.$

\emph{In foundations of quantum theory, }where\emph{\ }a\emph{\ }HV\ model
aims to reproduce in classical probability terms the statistical properties
of \emph{all} quantum observables on a Hilbert space $\mathcal{H}$, the
intention to mimic all the properties of quantum averages and quantum
correlations was realized via some additional functional assumptions \cite%
{1, 5, holevo} on a correspondence between random variables on $(\Omega ,%
\mathcal{F}_{\Omega })$ and quantum observables on $\mathcal{H}$. Under
these functional assumptions, for $\dim \mathcal{H}\geq 3,$ there does not
exist a noncontextual HV\ model reproducing the Hilbert space description of
all the joint von Neumann measurements (for details see section II).

\emph{In quantum information theory,} a HV model aims to reproduce \emph{only%
} the probabilistic description of a quantum correlation scenario\emph{\ }%
upon a state $\rho $ on a Hilbert space $\mathcal{H}_{1}\otimes \cdots
\otimes \mathcal{H}_{N}$, so that a noncontextual HV model is introduced 
\cite{10} directly via the noncontextual representation of all scenario
joint probabilities in classical probability terms. This HV representation
mimics by itself all the needed properties of quantum averages and quantum
correlations, so that, for a quantum correlation scenario, a setting of a
noncontextual HV model is not supplied \cite{10} by any additional
assumptions on a correspondence between random variables and modelled
quantum observables. As a result, for each dimension of a Hilbert space $%
\mathcal{H}_{1}\otimes \cdots \otimes \mathcal{H}_{N}$, there exist \cite{10}
quantum correlation scenarios admitting noncontextual HV\ models. Ever since
the arguments of Bell \cite{4, 3}, a noncontextual HV model for a quantum
correlation scenario is generally referred to as local and a LHV\ model, for
short (see section II for details).

From the point of view of quantum information applications, the existence of
LHV models for some quantum correlation scenarios questioned whether there
exists\emph{\ }a probability model which could reproduce the probabilistic
description of every quantum correlation scenario via random variables, each
depending only on a setting of the corresponding measurement at the
corresponding site, i. e. via "local" random variables.

As we proved in Ref. 10, the answer to this question is positive\emph{\ -- }%
the probabilistic description of \emph{every }quantum correlation scenario
does admit a new local probability model --\emph{\ a local quasi hidden
variable (LqHV) model, }where locality and the measure theory structure $%
(\Omega ,\mathcal{F}_{\Omega },\nu )$ inherent to a LHV model are preserved
but positivity of a simulation measure $\nu $ is dropped.

In a (deterministic \cite{12-}) LqHV model, specified by a measure space $%
(\Omega ,\mathcal{F}_{\Omega },\nu )$, all the joint probabilities and all
the quantum product averages of a correlation scenario are reproduced via
nonnegative values of a normalized real-valued measure $\nu $ and "local"
random variables on $(\Omega ,\mathcal{F}_{\Omega })$.

Moreover, we showed in Ref. 12 that the probabilistic description of every 
\emph{nonsignaling} \cite{10} correlation scenario, not necessarily quantum,
also admits a LqHV model.

Based on our results in Refs. 10, 12, we also introduced \cite{13} the
notion of \emph{a quasi-classical probability model }- a new general
probability model, which is specified in terms of a measure space $(\Omega ,%
\mathcal{F}_{\Omega },\nu )$ with a normalized real-valued measure $\nu $
and where a joint measurement with outcomes $(\lambda _{1},...,\lambda
_{n})\in \Lambda _{\theta _{1}}\times \cdots \times \Lambda _{\theta _{n}}$
and marginal measurements, represented by random variables $f_{\theta
_{i}}:\Omega \rightarrow \Lambda _{\theta _{i}},$ $i=1,...,n,$ is possible
if and only if 
\begin{equation}
\nu (f_{\theta _{1},...,\theta _{n}}^{-1}(B))\geq 0,\text{ \ \ }\forall
B\subseteq \Lambda _{\theta _{1}}\times \cdots \times \Lambda _{\theta _{n}},
\end{equation}%
where $f_{\theta _{1},...,\theta _{n}}=(f_{\theta _{1}},...,f_{\theta _{n}})$
and $f_{\theta _{1},...,\theta _{n}}^{-1}(B)=\{ \omega \in \Omega \mid
(f_{\theta _{1}}(\omega ),...,f_{\theta _{n}}(\omega ))\in B\}.$ This new
probability model reduces to the Kolmogorov probability model \cite{7, 17}
iff a real-valued measure $\nu $ is positive.

If a quasi-classical probability model is applied for the description of
quantum measurements, then, according to our terminology in Refs. 10, 12, we
refer to it as a qHV model.

In the present paper, we analyze further possibilities of the qHV approach 
\cite{12, 13} and prove that, for each Hilbert space, the Hilbert space
description of all the joint von Neumann measurements can be reproduced via%
\emph{\ }either of the two new qHV\ models, which we call as \emph{%
statistically noncontextual and context-invariant}.

\emph{In a statistically noncontextual model}\textbf{, }each quantum
observable $X$ is represented on $(\Omega ,\mathcal{F}_{\Omega })$ \emph{by
only one} random variable $f_{X}$ with values in the spectrum \textrm{sp}$X$
of this observable $X$ and the Kochen-Specker functional conditions\textbf{\ 
}\cite{5} are satisfied in average (for details see section II below).

\emph{In a context-invariant model}\textbf{, }a quantum observable $X$\ can
be modelled on\textbf{\ }$(\Omega ,\mathcal{F}_{\Omega })$ by a variety of
random variables satisfying the functional condition generally required in
quantum foundations (see relation (\ref{05}) below) but, in contrast to a
contextual model, each of these random variables equivalently models $X$\
under all joint von Neumann measurements, independently of their contexts,
in other words, the representation of $X$ by each of these random variables
is invariant with respect to a context of a joint von Neumann measurement.

For $\dim \mathcal{H\geq }4$, the HV\ versions of these models cannot exist
(see section II).

The\ proved existence of\emph{\ a context-invariant qHV\ model negates }$%
\emph{t}$\emph{he general opinion} that, via random variables satisfying the
functional condition required in quantum foundations, the Hilbert space
description of all the joint von Neumann measurements for $\dim \mathcal{H}%
\geq 3$ can be reproduced only contextually.

The paper is organized as follows.

In section II, we shortly review the settings of the HV models available in
the literature.

In section III, we recall the von Neumann formalism for the description of
ideal (projective) quantum measurements and the notion of the spectral
measure of a quantum observable.

In section IV, based on our generalization (appendix B) of some items of the
Kolmogorov extension theorem \cite{7} to the case of consistent
operator-valued measures, we prove (theorem 1) that all symmetrized products
of spectral measures admit the representation via the uniquely defined
self-adjoint operator-valued measure and the specific random variables on
the specially constructed measurable space $(\Lambda ,\mathcal{F}_{\Lambda
}) $.

In section V, we apply the new mathematical results of section IV to qHV
modelling (theorems 2, 3, corollaries 2, 3, propositions 2, 3) of the
Hilbert space description of the all joint von Neumann measurements for an
arbitrary Hilbert space.

In section VI, we discuss the main new results of the present paper.

\section{Preliminaries}

In this section, we shortly review settings of the HV\ models available in
quantum foundations and quantum information.

\emph{In foundations of quantum theory, }the first \emph{"no-go" theorem} on
non-existence of \emph{a noncontextual HV\ model }reproducing the
statistical properties of all quantum observables on $\mathcal{H}$ was
introduced \cite{1} by von Neumann in 1932.

However, analyzing \cite{4, 3} this problem in 1964 - 1966, Bell explicitly
constructed \cite{3} the noncontextual HV model for all quantum observables
of the qubit $(\dim \mathcal{H=}2)$ and argued \cite{3} that, though the
proof of the von Neumann "no-go" theorem is mathematically correct, the
setting of this theorem contains the linearity assumption which is
inconsistent with the quantum formalism and is not, in particular, fulfilled
in the specific HV\ model presented by him in Ref. 9.

In 1967, Kochen and Specker introduced \cite{5} a new setting for a
noncontextual HV model where a mapping $X\overset{\Phi }{\mapsto }f_{X},$ $%
f_{X}(\Omega )=\mathrm{sp}X,$ from the set of all quantum observables on $%
\mathcal{H}$ into the set of all random variables on a measurable space $%
(\Omega ,\mathcal{F}_{\Omega })$ was supplied by the physically motivated
functional condition%
\begin{eqnarray}
\Phi (\varphi \circ X) &=&\varphi \circ \Phi (X)  \label{03} \\
&\equiv &\varphi \circ f_{X}  \notag
\end{eqnarray}%
for all quantum observables $X$ on $\mathcal{H}$ and all Borel real-valued
functions $\varphi :\mathbb{R\rightarrow R}.$ Since quantum observables $%
X_{1},...,X_{n}$ mutually commute iff \cite{5---} there exist a quantum
observable $Z$ and Borel functions $\varphi _{i}:\mathbb{R\rightarrow R}$
such that all $X_{i}=\varphi _{i}(Z),$ $i=1,...,n,$ condition (\ref{03})
does imply linearity of a mapping $\Phi $:%
\begin{eqnarray}
\Phi (X_{1}+\cdots +X_{n}) &=&\Phi (X_{1})+\cdots +\Phi (X_{n})  \label{01}
\\
&\equiv &f_{X_{1}}+\cdots +f_{X_{n}}  \notag
\end{eqnarray}%
and also its multiplicativity%
\begin{equation}
(\Phi (X_{1}\cdot \ldots \cdot X_{n}))(\omega )=f_{X_{1}}(\omega )\cdot
\ldots \cdot f_{X_{n}}(\omega ),\text{ \ \ }\forall \omega \in \Omega ,
\label{02}
\end{equation}%
but -- \emph{only for mutually commuting quantum observables.}

Kochen and Specker proved \cite{5} that, for a Hilbert space of a dimension $%
\dim \mathcal{H}\geq 3,$ a noncontextual HV model supplied by the functional
conditions (\ref{01}), (\ref{02}) cannot exist.

For the case where, in a setting of a HV model, a quantum observable $X$ on $%
\mathcal{H}$ can be represented by a variety $\{f_{X}^{(\theta )},$ $\theta
\in \Theta _{X}\}$ of random variables on $(\Omega ,\mathcal{F}_{\Omega })$,
it was proved (proposition 1.4.2 in Ref. 6) that, for each Hilbert space%
\emph{\ }$\mathcal{H}\emph{,}$ there exist a measurable space $(\Omega ,%
\mathcal{F}_{\Omega })$ and a mapping $\Psi $ from a set of random variables
on $(\Omega ,\mathcal{F}_{\Omega })$ onto the set of all quantum observables
on $\mathcal{H}$ such that (i) the HV representation $\left \langle
X\right
\rangle _{\rho }=\langle f_{X}^{(\theta )}\rangle _{_{HV}}$ of the
quantum average of each observable $X$ in a state $\rho $ and (ii) the
functional condition 
\begin{eqnarray}
\Psi (\varphi \circ f_{X}^{(\theta )}) &=&\varphi \circ \Psi (f_{X}^{(\theta
)})  \label{05} \\
&=&\varphi \circ X,\text{ \ \ }\forall X,\text{ \ }\forall \varphi :\mathbb{%
R\rightarrow R},  \notag
\end{eqnarray}%
similar by its sense to condition (\ref{03}) are fulfilled for each of
random variables $\{f_{X}^{(\theta )},$ $\theta \in \Theta _{X}\}$
representing a quantum observable $X.$ If $\dim \mathcal{H}\geq 3,$ then
this mapping $\Psi $ cannot be injective due to the Kochen-Specker result 
\cite{5} and, therefore, the HV model specified by the above setting cannot
be \emph{noncontextual}. Furthermore, for $\dim \mathcal{H}\geq 3,$ this HV
model implies \cite{8--} the \emph{contextual} description of a joint von
Neumann measurement of mutually commuting quantum observables $%
X_{1},...,X_{n}$ -- in the sense that, under a joint von Neumann
measurement, an observable $X_{i}$ is, in general, represented by a random
variable\emph{\ }specific for a \emph{context} of this joint measurement.%
\emph{\ }Therefore, for a Hilbert space of a dimension $\dim \mathcal{H}\geq
3,$ the HV\ model specified by proposition 1.4.2 in Ref. 6 is \emph{%
contextual.}

Moreover, it is generally argued that this \emph{contextuality} is always
the case whenever a mapping $\Psi $ from a set of random variables on $%
(\Omega ,\mathcal{F}_{\Omega })$ onto the set of all quantum observables on $%
\mathcal{H}$ is non-injective. As we prove in section V, this opinion is
misleading.

\emph{In\ quantum\ information,\ }for some quantum correlation scenarios,
their noncontextual HV description is possible, while for others -- is
impossible. Since contextual HV modelling is possible for all $N$-partite
joint measurements on a $N$-partite quantum state, Bell presumed \cite{4, 3}
that, for an $N$-partite case, the only physical reason for non-existence of
a noncontextual HV model is \emph{quantum nonlocality }\cite{8-} -- in
contrast to \emph{locality }of\emph{\ }parties' measurements argued by
Einstein, Podolsky and Rosen (EPR) in Ref. 3.

Ever since this Bell's conjecture, a HV model for a quantum correlation
scenario is called local and is referred to as a \emph{LHV} model if, in
this model, each random variable modelling a party's measurement depends
only on a setting of this measurement at the corresponding site.

If all possible $N$-partite joint von Neumann measurements an $N$-partite
quantum state admit \cite{werner} a single LHV model, then this $N$-partite
state is also referred to as \emph{local.}

An arbitrary $N$-partite quantum correlation scenario does not need to admit
a LHV model. Also, an arbitrary $N$-partite quantum state does not need to
be local. However, not only every separable quantum state is local -- Werner 
\cite{werner} presented the example of a nonseparable (entangled) quantum
state on $\mathbb{C}^{d}\otimes \mathbb{C}^{d},$ $d\geq 2$, which is local
under all bipartite joint von Neumann measurements on this state. Moreover,
there also exist \cite{10, terhal, loubenets2, loubenets2a, loubenets2b}
nonseparable $N$-partite quantum states, that behave themselves as \emph{%
local} under every $N$-partite quantum correlation scenario with some
specific number $S_{n}\leq S_{n}^{(0)}$ of quantum measurements at each $n$%
-th site.

Furthermore, as we proved in Refs. 10, 12 \emph{every }$N$\emph{-partite
quantum correlation scenario admits} \emph{a local quasi hidden variable
(LqHV) model}, described in Introduction.

In view of the above analysis of different settings of HV\ models available
now in the literature and the new quasi hidden variable (qHV) approach,
developed in Refs. 10, 12 and outlined in Introduction, we put the following
questions important for both -- quantum foundations and quantum information.$%
\medskip $

\noindent \textbf{1.} Does every $N$-partite quantum state admit a LqHV
model?$\medskip $

\noindent \textbf{2.} Can the Hilbert space description of all the von
Neumann joint probabilities and product averages be reproduced via a qHV
model where to each quantum observables there corresponds only one random
variable and a mapping $X\overset{\Phi }{\mapsto }f_{X},$ $f_{X}(\Omega )=%
\mathrm{sp}X,$ satisfies \emph{in average }both\emph{\ -- }the von Neumann 
\cite{1} linearity assumption and the Kochen-Specker assumptions (\ref{03}),
(\ref{02}),\emph{\ }in other words,\emph{\ }satisfies for each quantum state 
$\rho $ the average relations%
\begin{eqnarray}
\langle \varphi \circ X\rangle _{\rho } &=&\langle \varphi \circ
f_{X}\rangle _{qHV},  \label{st1} \\
\left\langle X_{1}+\cdots +X_{n}\right\rangle _{\rho } &=&\left\langle
f_{X_{1}}+\cdots +f_{X_{n}}\right\rangle _{qHV},  \label{st1'} \\
\left\langle X_{1}\cdot \ldots \cdot X_{n}\right\rangle _{\rho }
&=&\left\langle f_{X_{1}}\cdot \ldots \cdot f_{X_{n}}\right\rangle _{qHV},%
\text{ \ \ }n\in \mathbb{N},  \label{st2}
\end{eqnarray}%
where relation (\ref{st1}) holds for all quantum observables $X$ and all
bounded Borel functions $\varphi :\mathbb{R\rightarrow R};$ relation (\ref%
{st1'}) -- for all bounded quantum observables $X_{1},\ldots ,X_{n}$ on $%
\mathcal{H}$ and relation (\ref{st2}) is satisfied only for mutually
commuting bounded quantum observables $X_{1},\ldots ,X_{n}?\medskip $

\noindent \textbf{3.} Does there exist a qHV model correctly reproducing all
the von Neumann joint probabilities and product expectations and where (i) a
quantum observable $X$ can be represented by a variety $\{f_{X}^{(\theta )},$
$\theta \in \Theta _{X}\}$ of random variables on $(\Omega ,\mathcal{F}%
_{\Omega })$ satisfying the functional condition (\ref{05}) required in
quantum foundations, but each of these random variables \emph{equivalently }%
represents a quantum observable $X$ under all joint von Neumann
measurements, \emph{independently} of their measurement contexts; (ii) the
average relations (\ref{st1}) - (\ref{st2}) hold with arbitrary
representatives $f_{X_{1}}^{(\theta _{1})},\ldots ,f_{X_{n}}^{(\theta _{n})}$
on $(\Omega ,\mathcal{F}_{\Omega })$ of quantum observables $X_{1},\ldots
,X_{n}$ on $\mathcal{H}?$

In what follows, we answer \emph{positively} to all of these questions and
refer to the qHV\ model specified in questions 2, 3 as statistically
noncontextual and context-invariant, respectively (for terminology, see also
Introduction).

Note that, for $\dim $\textbf{$\mathcal{H}$ }$\geq 4$\textbf{, }a
statistically noncontextual HV model and a context-invariant HV\ model for
all the joint von Neumann measurements cannot exist\ -- since, otherwise, in
view of isometric isomorphism of complex separable Hilbert spaces of the
same dimension, this would imply the existence of a LHV model for each
two-qubit state, but, as it is well known, this is not \cite{4} the case.

\section{Von Neumann measurements}

In the frame of the von Neumann formalism \cite{1}, states and observables
of a quantum system are described, correspondingly, by density operators $%
\rho $ and self-adjoint linear operators $X\ $on a complex separable Hilbert
space $\mathcal{H},$ possibly infinite dimensional.

Let $\mathcal{L}_{\mathcal{H}}^{(s)}$ be the real vector space of all
self-adjoint bounded linear operators on $\mathcal{H}.$ Equipped with the
operator norm, this vector space is Banach. Denote by $\mathfrak{X}_{%
\mathcal{H}}\supset \mathcal{L}_{\mathcal{H}}^{(s)}$ the set of all quantum
observables on $\mathcal{H}$, bounded and unbounded, and by $\mathrm{sp}%
X\subseteq \mathbb{R}$ the spectrum of a quantum observable $X.$

The probability that, under an ideal (errorless) measurement of a quantum
observable $X\in \mathfrak{X}_{\mathcal{H}}$ in a state $\rho ,$ an observed
value belongs to a Borel subset $B$ of $\mathbb{R}$ is given \cite{1, reed,
davies} by the expression 
\begin{equation}
\mathrm{tr}[\rho \mathrm{P}_{X}(B)],\text{ \ }B\in \mathcal{B}_{\mathbb{R}},
\label{1}
\end{equation}%
where $\mathcal{B}_{\mathbb{R}}$ is the Borel $\sigma $-algebra \cite%
{dunford} on $\mathbb{R}$ and $\mathrm{P}_{X}$ is the spectral measure of an
observable $X\in \mathfrak{X}_{\mathcal{H}},$ that is, the normalized
projection-valued measure $\mathrm{P}_{X}$ on $\mathcal{B}_{\mathbb{R}}$,
uniquely corresponding to an observable $X$ due to the spectral theorem \cite%
{1, reed} 
\begin{equation}
X=\dint\limits_{\mathbb{R}}x\mathrm{P}_{X}(\mathrm{d}x).  \label{2}
\end{equation}%
The values $\mathrm{P}_{X}(B),$ $B\in \mathcal{B}_{\mathbb{R}},$ $\mathrm{P}%
_{X}(\mathbb{R})=\mathbb{I}_{\mathcal{H}},$ of this measure are projections
on $\mathcal{H}$, satisfying the relations 
\begin{eqnarray}
\mathrm{P}_{X}(B_{1})\mathrm{P}_{X}(B_{2}) &=&\mathrm{P}_{X}(B_{2})\mathrm{P}%
_{X}(B_{1})=\mathrm{P}_{X}(B_{1}\cap B_{2}),\text{ \ \ }B_{1},B_{2}\in 
\mathcal{B}_{\mathbb{R}},  \label{3} \\
\mathrm{P}_{X}(B) &=&0,\text{\ \ }B\notin \mathcal{B}_{\mathbb{R}}\cap 
\mathrm{sp}X.  \notag
\end{eqnarray}%
For each $X\in \mathfrak{X}_{\mathcal{H}},$ its spectrum $\mathrm{sp}X\in 
\mathcal{B}_{\mathbb{R}}.$ Due to the second relation in (\ref{3}), we
further consider the spectral measure $\mathrm{P}_{X},$ $X\in \mathfrak{X}_{%
\mathcal{H}},$ only on the trace $\sigma $-algebra 
\begin{equation}
\mathcal{B}_{\mathrm{sp}X}:=\mathcal{B}_{\mathbb{R}}\cap \mathrm{sp}X.
\end{equation}

The measure $\mathrm{P}_{X}$ is $\sigma $-additive in the strong operator
topology \cite{davies, 15} in $\mathcal{L}_{\mathcal{H}}^{(s)}$, that is:%
\begin{equation}
\lim_{n\rightarrow \infty }\left \Vert \mathrm{P}_{X}(\cup _{i=1}^{\infty
}B_{i})\psi -\sum_{i=1}^{n}\mathrm{P}_{X}(B_{i})\psi \right \Vert _{\mathcal{%
H}}=0  \label{4}
\end{equation}%
for all $\psi \in \mathcal{H}$ and all countable collections $\{B_{i}\}$ of
mutually disjoint sets in $\mathcal{B}_{\mathrm{sp}X}.$

\begin{remark}
In this article, we follow the terminology of Ref. 27. Namely, let $%
\mathfrak{B}$ be a Banach space and $\mathcal{F}_{\Lambda }$ be an algebra
of subsets of a set $\Lambda .$ We refer to an additive set function $%
\mathfrak{m}:\mathcal{F}_{\Lambda }\rightarrow \mathfrak{B}$ as a $\mathfrak{%
B}$-valued (finitely additive) measure on $\mathcal{F}_{\Lambda }$. If a
measure $\mathfrak{m}$ on $\mathcal{F}_{\Lambda }$ is $\sigma $-additive in
some topology in $\mathfrak{B},$ then we specify this in addition.
\end{remark}

An ideal measurement (\ref{1}) of a quantum observable $X$ in a state $\rho $
is generally referred to as \emph{the von Neumann measurement.}

The joint von Neumann measurement of several quantum observables $%
X_{1},...,X_{n}\in \mathfrak{X}_{\mathcal{H}}$ is possible \cite{1, davies}
if and only if all values of their spectral measures mutually commute: 
\begin{equation}
\lbrack \mathrm{P}_{X_{i_{1}}}(B_{i_{1}}),\mathrm{P}%
_{X_{i_{2}}}(B_{i_{2}})]=0,\text{ \ \ }B_{i}\in \mathcal{B}_{\mathrm{sp}%
X_{i}},\text{ \ \ }i=1,...,n.  \label{com}
\end{equation}%
and is described in this case by the projection-valued product measure \cite%
{1, davies}%
\begin{equation}
\mathrm{P}_{X_{1},...,X_{n}}(B):=\dint \limits_{(x_{1},...,x_{n})\in B}%
\mathrm{P}_{X_{1}}(\mathrm{d}x_{1})\cdot ...\cdot \mathrm{P}_{X_{n}}(\mathrm{%
d}x_{n}),\text{\ \ }B\in \mathcal{B}_{\mathrm{sp}X_{1}\times \cdots \times 
\mathrm{sp}X_{n}},  \label{5}
\end{equation}%
which is normalized $\mathrm{P}_{X_{1},...,X_{n}}(\mathrm{sp}X_{1}\times
\cdots \times \mathrm{sp}X_{n})=\mathbb{I}_{\mathcal{H}}$ and defined on the 
$\sigma $-algebra%
\begin{eqnarray}
\mathcal{B}_{\mathrm{sp}X_{1}\times \cdots \times \mathrm{sp}X_{n}} &:&=%
\mathcal{B}_{\mathbb{R}^{n}}\cap (\mathrm{sp}X_{1}\times \cdots \times 
\mathrm{sp}X_{n}) \\
&=&\mathcal{B}_{\mathrm{sp}X_{1}}\otimes \cdots \otimes \mathcal{B}_{\mathrm{%
sp}X_{n}}  \notag
\end{eqnarray}%
of Borel subsets of $\mathrm{sp}X_{1}\times \cdots \times \mathrm{sp}X_{n}.$

For bounded quantum observables $X_{1},...,X_{n}\in \mathcal{L}_{\mathcal{H}%
}^{(s)}$, condition (\ref{com}) is equivalent to mutual commutativity $%
[X_{i_{1}},X_{i_{2}}]=0,$ $i=1,...,n$. of these observables. Therefore, for
short, we further refer to arbitrary quantum observables $X_{1},...,X_{n}$,
bounded or unbounded, as mutually commuting if their spectral measures
satisfy condition (\ref{com}). The measure $\mathrm{P}_{X_{1},...,X_{n}}$ is
also referred\ \cite{15} to as the \emph{joint spectral measure} of mutually
commuting quantum observables $X_{1},...,X_{n}.$

The expression 
\begin{equation}
\mathrm{tr}[\rho \mathrm{P}_{X_{1},...,X_{n}}(B_{1}\times \ldots \times
B_{n})]=\mathrm{tr}[\rho (\mathrm{P}_{X_{1}}(B_{1})\cdot ...\cdot \mathrm{P}%
_{X_{n}}(B_{n}))]  \label{6_}
\end{equation}%
defines the probability that the observed values of mutually commuting
quantum observables $X_{1},...,X_{n}\ $are in sets $B_{1}\in \mathcal{B}_{%
\mathrm{sp}X_{1}},...,B_{n}\in \mathcal{B}_{\mathrm{sp}X_{n}},$ respectively.

\section{Symmetrized products of spectral measures}

In this section, for our further consideration in section V, we introduce a
new operator-valued measure induced by the symmetrized product of spectral
measures of quantum observables and prove the extension theorem for the
consistent family of these operator-valued measures.

For an $n$-tuple $(X_{1},...,X_{n})$ of arbitrary mutually non-equal quantum
observables $X_{1},...,X_{n}$ $\in \mathfrak{X}_{\mathcal{H}},$ let $%
\mathcal{F}_{\mathrm{sp}X_{1}\times \cdots \times \mathrm{sp}X_{n}}$ be the 
\emph{product algebra }on $\mathrm{sp}X_{1}\times \cdots \times \mathrm{sp}%
X_{n}\subseteq \mathbb{R}^{n},$ that is, the algebra generated by all
rectangles $B_{1}\times \cdots \times B_{n}$ $\subseteq \mathrm{sp}%
X_{1}\times \cdots \times \mathrm{sp}X_{n}$ with measurable sides $B_{i}\in 
\mathcal{B}_{\mathrm{sp}X_{i}}.$

Let 
\begin{equation}
\mathcal{P}_{(X_{1},...,X_{n})}:\mathcal{F}_{\mathrm{sp}X_{1}\times \cdots
\times \mathrm{sp}X_{n}}\rightarrow \mathcal{L}_{\mathcal{H}}^{(s)},\text{ \
\ \ \ \ \ }\mathcal{P}_{(X_{1},...,X_{n})}(\mathrm{sp}X_{1}\times \cdots
\times \mathrm{sp}X_{n})=\mathbb{I}_{\mathcal{H}},  \label{me}
\end{equation}%
be the normalized (finitely additive) product measure on $\mathcal{F}_{%
\mathrm{sp}X_{1}\times \cdots \times \mathrm{sp}X_{n}},$ defined uniquely
via its representation%
\begin{equation}
\mathcal{P}_{(X_{1},...,X_{n})}(B_{1}\times \cdots \times B_{n})=\frac{1}{n!}%
\left \{ \mathrm{P}_{X_{1}}(B_{1})\cdot \ldots \cdot \mathrm{P}%
_{X_{n}}(B_{n})\right \} _{\mathrm{sym}}  \label{9_}
\end{equation}%
on all rectangles $B_{1}\times \cdots \times B_{n}$ with $B_{i}\in \mathcal{B%
}_{\mathrm{sp}X_{i}}$. The values of the measure $\mathcal{P}%
_{(X_{1},...,X_{n})}$ are self-adjoint bounded linear operators on $\mathcal{%
H}.$ Here, notation $\{Z_{1}\cdot \ldots \cdot Z_{n}\}_{\mathrm{sym}}$ means
the sum constituting the symmetrization of the operator product $Z_{1}\cdot
\ldots \cdot Z_{n},$ where $Z_{i}\in \mathcal{L}_{\mathcal{H}}^{(s)}$, with
respect to all permutations of its factors.

For a collection $\{X_{1},...,X_{n}\}$ of mutually commuting quantum
observables, the measure $\mathcal{P}_{(X_{1},...,X_{n})}$ is
projection-valued and constitutes the restriction to the product algebra $%
\mathcal{F}_{\mathrm{sp}X_{1}\times \cdots \times \mathrm{sp}X_{n}}$ of the
joint spectral measure $\mathrm{P}_{X_{1},...,X_{n}}$ defined by relation (%
\ref{5}) on the $\sigma $-algebra $\mathcal{B}_{\mathrm{sp}X_{1}\times
\cdots \times \mathrm{sp}X_{n}}$.

If each observable $X_{i}$ in a collection $\{X_{1},...,X_{n}\}$ is bounded,
i.e. $X_{i}\in \mathcal{L}_{\mathcal{H}}^{(s)}$, and, moreover, has only a
discrete spectrum $\mathrm{sp}X_{i}=\{x_{i}^{(k)}\in \mathbb{R},$ $%
k=1,...,K_{X_{i}}<\infty \},$ where each $x_{i}^{(k)}$ is an eigenvalue of $%
X_{i}$, then the product algebra $\mathcal{F}_{\mathrm{sp}X_{1}\times \cdots
\times \mathrm{sp}X_{n}}$ is finite and coincides with the $\sigma $-algebra 
$\mathcal{B}_{\mathrm{sp}X_{1}\times \cdots \times \mathrm{sp}X_{n}}$ while
the product measure $\mathcal{P}_{(X_{1},...,X_{n})}$ takes the form 
\begin{equation}
\mathcal{P}_{(X_{1},...,X_{n})}(F):=\frac{1}{n!}\dsum%
\limits_{(x_{1},...,x_{n})\in F}\left \{ \mathrm{P}_{X_{1}}(\{x_{1}\})\cdot
\ldots \cdot \mathrm{P}_{X_{n}}(\{x_{n}\})\right \} _{\mathrm{sym}}
\label{9}
\end{equation}%
for all $F\in \mathcal{F}_{\mathrm{sp}X_{1}\times \cdots \times \mathrm{sp}%
X_{n}}.$

Consider the family%
\begin{equation}
\{ \mathcal{P}_{(X_{1},...,X_{n})}\mid \{X_{1},...,X_{n}\} \subset \mathfrak{%
X}_{\mathcal{H}},\text{ }n\in \mathbb{N}\}  \label{12}
\end{equation}%
of all normalized $\mathcal{L}_{\mathcal{H}}^{(s)}$-valued measures (\ref{me}%
), each specified by a tuple $(X_{1},...,X_{n})$ of mutually non-equal
quantum observables on $\mathcal{H}$. These measures satisfy the following
consistency relations proved in appendix A.

\begin{lemma}
For every collection $\{X_{1},...,X_{n}\}$ $\subset \mathfrak{X}_{\mathcal{H}%
},$ $n\in \mathbb{N},$ of quantum observables on $\mathcal{H},$ the relation 
\begin{eqnarray}
\mathcal{P}_{(X_{1},...,X_{n})}(B_{1}\times \cdots \times B_{n}) &=&\mathcal{%
P}_{(X_{i_{1}},...,X_{i_{_{n}}})}(B_{i_{1}}\times \cdots \times B_{i_{n}}),
\label{13} \\
B_{1} &\in &\mathcal{B}_{\mathrm{sp}X_{1}},...,B_{n}\in \mathcal{B}_{\mathrm{%
sp}X_{n}},  \notag
\end{eqnarray}%
holds for all permutations $\binom{1,,...,n}{i_{1},...,i_{n}}$ of indexes
and the relation 
\begin{eqnarray}
&&\mathcal{P}_{(X_{1},...,X_{_{n}})}\left( \{\text{ }(x_{1},...,x_{n})\in 
\mathrm{sp}X_{1}\times \cdots \times \mathrm{sp}X_{n}\mid
(x_{i_{1}},...,x_{i_{k}})\in F\}\right)  \label{14} \\
&=&\mathcal{P}_{(X_{i_{1}},...,X_{i_{k}})}\left( F\right) ,\text{ \ \ \ \ \ }%
F\in \mathcal{F}_{\mathrm{sp}X_{i_{1}}\times \cdots \times \mathrm{sp}%
X_{i_{k}}},  \notag
\end{eqnarray}%
is fulfilled for each $\{X_{i_{1}},...,X_{i_{k}}\}\subseteq
\{X_{1},...,X_{n}\}.\medskip $
\end{lemma}

Note that, for the operator-valued measures $\{\mathcal{P}%
_{(X_{1},...,X_{_{n}})}\}$, relations (\ref{13}), (\ref{14}) are quite
similar by their form to the Kolmogorov consistency conditions \cite{7, 17}
for a family%
\begin{equation}
\{\mu _{(t_{1},...,t_{n})}:\mathcal{B}_{\mathbb{R}^{n}}\rightarrow \lbrack
0,1]\text{ }\mid \{t_{1},...,t_{n}\}\subset T,\ \ n\in \mathbb{N}\}
\label{15}
\end{equation}%
of probability measures $\mu _{(t_{1},...,t_{n})}$, each specified by a
tuple $(t_{1},...,t_{n})$ of mutually non-equal elements in an index set $T.$

In view of this similarity and for our further consideration, in appendix B,
we generalize to the case of consistent operator-valued measures some items
of the Kolmogorov extension theorem \cite{7, 17} for consistent probability
measures (\ref{15}).

\subsection{The extension theorem}

Denote by $\Lambda :=\dprod \limits_{X\in \text{ }\mathfrak{X}_{\mathcal{H}}}%
\mathrm{sp}X$ the Cartesian product of the spectrums of all the quantum
observables on $\mathcal{H}.$ By its definition \cite{dunford}, $\Lambda $
is the set of all real-valued functions 
\begin{equation}
\lambda :\mathfrak{X}_{\mathcal{H}}\rightarrow \cup _{X\in \mathfrak{X}_{%
\mathcal{H}}}\mathrm{sp}X
\end{equation}%
with values $\lambda (X)\equiv \lambda _{X}\in \mathrm{sp}X.$

Let the random variable $\pi _{(X_{1},...,X_{n})}:\Lambda \rightarrow 
\mathrm{sp}X_{1}\times \cdots \times \mathrm{sp}X_{n}$ be the canonical
projection on $\Lambda :$%
\begin{eqnarray}
\pi _{(X_{1},...,X_{n})}(\lambda ) &:&=\left( \pi _{X_{1}}(\lambda ),...,\pi
_{X_{n}}(\lambda )\right) ,  \label{canon} \\
\pi _{X}(\lambda ) &:&=\lambda _{X}\in \mathrm{sp}X.  \notag
\end{eqnarray}%
The set%
\begin{equation}
\mathcal{A}_{\Lambda }=\left \{ \pi _{(X_{1},...,X_{n})}^{-1}(F)\subseteq
\Lambda \mid F\in \mathcal{F}_{\mathrm{sp}X_{1}\times \cdots \times \mathrm{%
sp}X_{n}},\text{\ \ }\{X_{1},...,X_{n}\} \subset \mathfrak{X}_{\mathcal{H}},%
\text{\ \ }n\in \mathbb{N}\right \}  \label{sigma}
\end{equation}%
of all cylindrical subsets of $\Lambda $ of the form 
\begin{equation}
\pi _{(X_{1},...,X_{n})}^{-1}(F):=\left \{ \lambda \in \Lambda \mid (\pi
_{X_{1}}(\lambda ),...,\pi _{X_{n}}(\lambda ))\in F\right \} ,  \label{f}
\end{equation}%
constitutes an algebra on $\Lambda $ (proposition III.11.18 in Ref. 25).

Due to lemma 3 proved in appendix B and generalizing some items of the
Kolmogorov extension theorem to the case of consistent operator-valued
measures, we have the following statement.

\begin{theorem}[The extension theorem]
Let $\mathcal{H}$ be a complex separable Hilbert space. For family (\ref{12}%
) of finitely additive $\mathcal{L}_{\mathcal{H}}^{(s)}$-valued measures $%
\mathcal{P}_{(X_{1},...,X_{n})}:\mathcal{F}_{\mathrm{sp}X_{1}\times \cdots
\times \mathrm{sp}X_{n}}\rightarrow \mathcal{L}_{\mathcal{H}}^{(s)},$ there
exists a unique normalized finitely additive $\mathcal{L}_{\mathcal{H}%
}^{(s)} $-valued measure 
\begin{equation}
\mathbb{M}:\mathcal{A}_{\Lambda }\rightarrow \mathcal{L}_{\mathcal{H}%
}^{(s)},\ \ \ \ \ \ \mathbb{M(}\Lambda )=\mathbb{I}_{\mathcal{H}},
\label{19}
\end{equation}%
on $(\Lambda ,\mathcal{A}_{\Lambda })$ such that%
\begin{equation}
\mathcal{P}_{(X_{1},...,X_{n})}(F)=\mathbb{M}\left( \pi
_{(X_{1},...,X_{n})}^{-1}(F)\right) ,\text{ \ \ }F\in \mathcal{F}_{\mathrm{sp%
}X_{1}\times \cdots \times \mathrm{sp}X_{n}},  \label{20}
\end{equation}%
in particular,%
\begin{eqnarray}
\frac{1}{n!}\left \{ \mathrm{P}_{X_{1}}(B_{1})\cdot \ldots \cdot \mathrm{P}%
_{X_{n}}(B_{n})\right \} _{\mathrm{sym}} &=&\mathbb{M(}\pi
_{X_{1}}^{-1}(B_{1})\cap \cdots \cap \pi _{X_{n}}^{-1}(B_{n})),  \label{21}
\\
B_{i} &\in &\mathcal{B}_{\mathrm{sp}X_{i}},\text{ }i=1,...,n,  \notag
\end{eqnarray}%
for all collections $\{X_{1},...,X_{n}\} \subset \mathfrak{X}_{\mathcal{H}},$
$n\in \mathbb{N},$ of quantum observables on $\mathcal{H}.$
\end{theorem}

\begin{proof}
The family $\{\mathcal{P}_{(X_{1},...,X_{n})}\}$ represents a particular
example of a general family (\ref{4_}) if, in the latter, we replace 
\begin{eqnarray}
T &\rightarrow &\mathfrak{X}_{\mathcal{H}},\text{ \ \ }\Lambda
_{t}\rightarrow \mathrm{sp}X,\text{ \ \ }\widetilde{\Lambda }\rightarrow
\Lambda \\
\mathcal{F}_{t} &\rightarrow &\mathcal{B}_{\mathrm{sp}X},\text{ \ \ \ }%
\mathcal{F}_{\Lambda _{t_{1}}\times \cdots \times \Lambda
_{t_{n}}}\rightarrow \mathcal{F}_{\mathrm{sp}X_{1}\times \cdots \times 
\mathrm{sp}X_{n}}.  \notag
\end{eqnarray}%
Moreover, by lemma 1, the family $\{\mathcal{P}_{(X_{1},...,X_{n})}\}$
satisfies the consistency conditions (\ref{5_}), (\ref{6__}). Therefore,
representation (\ref{20}) follows explicitly from relation (\ref{8_}) in
lemma 3 of appendix B. This proves the statement. \medskip
\end{proof}

Theorem 1 allows us to express all real-valued measures%
\begin{equation}
\mathrm{tr}[\rho \mathcal{P}_{(X_{1},...,X_{n})}(\cdot )],\text{ \ \ }%
\{X_{1},...,X_{n}\} \subset \mathfrak{X}_{\mathcal{H}},\text{ \ }n\in 
\mathbb{N},
\end{equation}%
for a quantum state $\rho $ on $\mathcal{H}$ via a single real-valued
measure on the algebra $\mathcal{A}_{\Lambda }.$

\begin{proposition}
Let $\mathcal{H}$ be a complex separable Hilbert space and $\{ \mathcal{P}%
_{(X_{1},...,X_{n})}\}$ be family (\ref{12}) of $\mathcal{L}_{\mathcal{H}%
}^{(s)}$-valued measures (\ref{me}). To every state $\rho $ on $\mathcal{H}$%
, there corresponds ($\rho \overset{\mathfrak{R}}{\mapsto }\mu _{\rho })$ a
unique normalized finitely additive real-valued measure%
\begin{equation}
\mu _{\rho }:\mathcal{A}_{\Lambda }\rightarrow \mathbb{R},\text{ \ \ }\mu
_{\rho }(\Lambda )=1,  \label{22}
\end{equation}%
on $(\Lambda ,\mathcal{A}_{\Lambda })$ such that 
\begin{equation}
\mathrm{tr}[\rho \mathcal{P}_{(X_{1},...,X_{n})}(F)]=\mu _{\rho }\left( \pi
_{(X_{1},...,X_{n})}^{-1}(F)\right) ,\text{ \ \ }F\in \mathcal{F}_{\mathrm{sp%
}X_{1}\times \cdots \times \mathrm{sp}X_{n}},  \label{23}
\end{equation}%
in particular,%
\begin{eqnarray}
\frac{1}{n!}\mathrm{tr}[\rho \{ \mathrm{P}_{X_{1}}(B_{1})\cdot \ldots \cdot 
\mathrm{P}_{X_{n}}(B_{n})\}_{\mathrm{sym}}] &=&\mu _{\rho }\left( \pi
_{X_{1}}^{-1}(B_{1})\cap \cdots \cap \pi _{X_{n}}^{-1}(B_{n})\right) ,
\label{24} \\
B_{i} &\in &\mathcal{B}_{\mathrm{sp}X_{i}},\text{ \ }i=1,...,n,  \notag
\end{eqnarray}%
for all collections $\{X_{1},...,X_{n}\} \subset \mathfrak{X}_{\mathcal{H}},$
$n\in \mathbb{N},$ of quantum observables on $\mathcal{H}$. If $\rho _{j}%
\overset{\mathfrak{R}}{\mapsto }\mu _{\rho _{j}},$ $j=1,...,m,$ then 
\begin{equation}
\sum \alpha _{j}\rho _{j}\overset{\mathfrak{R}}{\mapsto }\sum \alpha _{j}\mu
_{\rho _{j}},\text{ \ \ \ }\alpha _{j}>0,\text{ \ }\sum \alpha _{j}=1.
\label{ma}
\end{equation}
\end{proposition}

\begin{proof}
For a state $\rho $ on $\mathcal{H},$ representation (\ref{20}) implies%
\begin{equation}
\mathrm{tr}[\rho \mathcal{P}_{(X_{1},...,X_{n})}(F)]=\mathrm{tr}[\rho 
\mathbb{M}\left( \pi _{(X_{1},...,X_{n})}^{-1}(F)\right) ]  \label{25}
\end{equation}%
for all sets $F\in \mathcal{F}_{\mathrm{sp}X_{1}\times \cdots \times \mathrm{%
sp}X_{n}}$ and all finite collections $\{X_{1},...,X_{n}\} \subset \mathfrak{%
X}_{\mathcal{H}}.$ Introduce on the algebra $\mathcal{A}_{\Lambda }$ the set
function 
\begin{equation}
\mu _{\rho }\left( A\right) :=\mathrm{tr}[\rho \mathbb{M}\left( A\right) ],%
\text{ \ \ }A\in \mathcal{A}_{\Lambda },  \label{26}
\end{equation}%
defined uniquely $\rho \overset{\mathfrak{R}}{\mapsto }\mu _{\rho }$ to each
state $\rho .$ Since $\mathbb{M}$ is a normalized (finitely additive)
measure on the algebra $\mathcal{A}_{\Lambda }$, also, $\mu _{\rho }$ is a
normalized (finitely additive) measure on $\mathcal{A}_{\Lambda }.$
Moreover, since $\mathbb{M}$ is a unique measure on $\mathcal{A}_{\Lambda },$
satisfying representation (\ref{20}), the measure $\mu _{\rho },$ uniquely
defined to each state $\rho $ by relation (\ref{26}), is also a unique
normalized real-valued measure on $\mathcal{A}_{\Lambda }$, satisfying
representation (\ref{25}), hence, (\ref{23}) and (\ref{24}).

Further, if $\rho _{j}\overset{\mathfrak{R}}{\mapsto }\mu _{\rho _{j}},$
then, due to definition (\ref{26}) of this mapping, the measure $\sum \alpha
_{j}\mu _{\rho _{j}},$ where $\alpha _{j}>0$ and $\sum \alpha _{j}=1,$ is a
unique real-valued measure, corresponding to the state $\sum \alpha _{j}\rho
_{j}$ due to (\ref{26}) and satisfying (\ref{23}), (\ref{24}). This
completes the proof.\medskip
\end{proof}

For bounded quantum observables, proposition 1 implies the following
representation.

\begin{corollary}
Let $\mathcal{H}$ be a complex separable Hilbert space. The representation 
\begin{equation}
\frac{1}{n!}\mathrm{tr}[\rho \{X_{1}\cdot \ldots \cdot X_{n}\}_{\mathrm{sym}%
}]=\dint \limits_{\Lambda }\pi _{_{X_{1}}}(\lambda )\cdot \ldots \cdot \pi
_{_{X_{n}}}(\lambda )\text{ }\mu _{\rho }\left( \mathrm{d}\lambda \right)
\label{po}
\end{equation}%
holds for all states $\rho $ and all finite collections $\{X_{1},\ldots
,X_{n}\}$ of bounded quantum observables on $\mathcal{H}.$
\end{corollary}

\begin{proof}
For bounded quantum observables $X_{1},...,X_{n}$, the operator $\rho
\{X_{1}\cdot \ldots \cdot X_{n}\}_{\mathrm{sym}}$ is trace class, so that
the quantum average $\mathrm{tr}[\rho \{X_{1}\cdot \ldots \cdot X_{n}\}_{%
\mathrm{sym}}]<\infty $ exists for all states $\rho $. Combining (\ref{24})
and (\ref{2}), we derive (\ref{po}).
\end{proof}

\section{Quasi hidden variable (qHV) modelling}

As we discussed in section III, the joint von Neumann measurement of quantum
observables $X_{1},...,X_{n}$ on $\mathcal{H}$ is possible if and only if
these observables mutually commute and is described in this case by the
joint spectral measure $\mathrm{P}_{X_{1},...,X_{n}}$ defined by (\ref{5}).
If a joint von Neumann measurement of mutually commuting quantum observables 
$X_{1},...,X_{n}$ is performed on a quantum system in a state $\rho $ on $%
\mathcal{H},$ then the expression 
\begin{equation}
\mathrm{tr}[\rho \mathrm{P}_{X_{1},...,X_{n}}(B)],\text{ \ \ }B\in \mathcal{B%
}_{\mathrm{sp}X_{1}\times \cdots \times \mathrm{sp}X_{n}},  \label{vn'}
\end{equation}%
gives the probability that these observables take values $x_{1},...x_{n}$
such that $(x_{1},...x_{n})\in B.$ In particular, 
\begin{equation}
\mathrm{tr}[\rho \{ \mathrm{P}_{X_{1}}(B_{1})\cdot \ldots \cdot \mathrm{P}%
_{X_{n}}(B_{n})\}]  \label{vn}
\end{equation}%
is the probability that the observed values of $X_{1},...,X_{n}\ $are in
sets $B_{1}\in \mathcal{B}_{\mathrm{sp}X_{1}},...,B_{n}\in \mathcal{B}_{%
\mathrm{sp}X_{n}},$ respectively.

For a collection $\{X_{1},...,X_{n}\}$ of mutually commuting quantum
observables, the measure $\mathcal{P}_{(X_{1},...,X_{n})}$ discussed in
theorem 1 coincides with the restriction of the joint spectral measure $%
\mathrm{P}_{X_{1},...,X_{n}}$ to the algebra $\mathcal{F}_{\mathrm{sp}%
X_{1}\times \cdots \times \mathrm{sp}X_{n}}$.

In the following sections, this allows us to analyze (theorems 2, 3) a
possibility of modelling of the Hilbert space description (\ref{vn}) of all
joint von Neumann measurements \emph{in measure theory terms}.

\subsection{A statistically noncontextual qHV model}

In this section, we formulate and prove the statements (theorem 2, corollary
2, propositions 2, 3) which immediately give the positive answers to
questions 2 and 1, formulated in section II.

For our below consideration, we recall that if quantum observables $%
X_{1},...,X_{n}$ mutually commute, then, for each Borel function $\psi :%
\mathbb{R}^{n}\rightarrow \mathbb{R},$ the notation $\psi (X_{1},...,X_{n})$
means \cite{15} the quantum observable 
\begin{equation}
\psi (X_{1},...,X_{n}):=\dint \limits_{\mathbb{R}^{n}}\psi (x_{1},...,x_{n})%
\text{ }\mathrm{P}_{X_{1}}(\mathrm{d}x_{1})\cdot \ldots \cdot \mathrm{P}%
_{X_{n}}(\mathrm{d}x_{n}).  \label{qww}
\end{equation}%
If a real-valued function $\psi $ is bounded, then the quantum observable $%
\psi (X_{1},...,X_{n})$ is also bounded.

\begin{theorem}
Let $\mathcal{H}$ be an arbitrary complex separable Hilbert space. There
exist:\smallskip \newline
(i) a set $\Omega $ and an algebra $\mathcal{F}_{\Omega }$ of subsets of $%
\Omega ;$\smallskip \newline
(ii) a one-to-one mapping $\Phi :\mathfrak{X}_{\mathcal{H}}\rightarrow 
\mathfrak{F}_{(\Omega ,\mathcal{F}_{\Omega })}$ from the set $\mathfrak{X}_{%
\mathcal{H}}$ of all quantum observables on $\mathcal{H}$ into the set $%
\mathfrak{F}_{(\Omega ,\mathcal{F}_{\Omega })}$ of all random variables on $%
(\Omega ,\mathcal{F}_{\Omega })$, with values $f_{X}:=\Phi (X)$ satisfying
the spectral correspondence rule $f_{X}(\Omega )=\mathrm{sp}X;$\medskip \ 
\newline
such that, to each quantum state $\rho $ on $\mathcal{H},$ there corresponds 
$(\rho \overset{\mathfrak{R}}{\mapsto }\nu _{\rho })$ a unique normalized
real-valued measure $\nu _{\rho }$ on $(\Omega ,\mathcal{F}_{\Omega }),$
satisfying the relation \cite{16-}%
\begin{eqnarray}
\mathrm{tr}[\rho \mathrm{P}_{X_{1},...,X_{n}}(F)] &=&\nu _{\rho }\left(
f_{(X_{1},...X_{n})}^{-1}(F)\right) ,\text{ \ \ }F\in \mathcal{F}_{\mathrm{sp%
}X_{1}\times \cdots \times \mathrm{sp}X_{n}},  \label{28'} \\
f_{(X_{1},...,X_{n})} &:&=(f_{X_{1}},...,f_{X_{n}}),  \notag
\end{eqnarray}%
in particular, 
\begin{eqnarray}
\mathrm{tr}[\rho \{ \mathrm{P}_{X_{1}}(B_{1})\cdot \ldots \cdot \mathrm{P}%
_{X_{n}}(B_{n})\}] &=&\nu _{\rho }\left( \text{ }f_{X_{1}}^{-1}(B_{1})\cap
\cdots \cap f_{X_{n}}^{-1}(B_{n})\right)  \label{28''} \\
&\equiv &\dint \limits_{\Omega }\chi _{f_{X_{1}}^{-1}(B_{1})}(\omega )\cdot
\ldots \cdot \chi _{f_{X_{n}}^{-1}(B_{n})}(\omega )\nu _{\rho }(\text{%
\textrm{d}}\omega ),  \notag \\
B_{i} &\in &\mathcal{B}_{\mathrm{sp}X_{i}},\text{ \ \ }i=1,...,n,  \notag
\end{eqnarray}%
for all collections $\{X_{1},...,X_{n}\},$ $n\in \mathbb{N},$ of mutually
commuting quantum observables on $\mathcal{H}$.\textbf{\ }In (\ref{28'}), $%
P_{X_{1},...,X_{n}}$ is the joint spectral measure (\ref{5}). If $\rho _{j}%
\overset{\mathfrak{R}}{\mapsto }\nu _{\rho _{j}},$ $j=1,...,m<\infty ,$ then 
$\sum \alpha _{j}\rho _{j}\overset{\mathfrak{R}}{\mapsto }\sum \alpha
_{j}\nu _{\rho _{j}},$ for all $\alpha _{j}>0,$ $\sum \alpha _{j}=1.$
\end{theorem}

\begin{proof}
In order to prove the existence point of theorem 2, let us take the
measurable space $(\Lambda ,\mathcal{A}_{\Lambda })$ and the random
variables $\pi _{X}(\lambda )=\lambda _{X}\in \mathrm{sp}X,$ $X\in \mathfrak{%
X}_{\mathcal{H}}$, on this space, which are specified in section IV.A. Since 
$\pi _{X_{1}}\neq \pi _{X_{2}}$ $\Leftrightarrow $ \ $X_{1}\neq X_{2}$ and $%
\pi _{X}(\Lambda )=\mathrm{sp}X,$ the set $\{ \pi _{X}\mid X\in \mathfrak{X}%
_{\mathcal{H}}\}$ of random variables is put into the one-to-one
correspondence to the set $\mathfrak{X}_{\mathcal{H}}$ of all quantum
observables on $\mathcal{H}$ and, for each random variable $\pi _{X},$ the
spectral correspondence rule $\pi _{X}(\Lambda )=\mathrm{sp}X$ is fulfilled.

Furthermore, by proposition 1, to each quantum state $\rho $ on $\mathcal{H}%
, $ there corresponds a unique normalized real-valued measures $\mu _{\rho }$%
, defined on the algebra $\mathcal{A}_{\Lambda }$ and satisfying
representations (\ref{23}), (\ref{24}). For a collection $%
\{X_{1},...,X_{n}\} $ of mutually commuting quantum observables, the measure 
$\mathcal{P}_{(X_{1},...,X_{n})}$ in the left-hand sides of (\ref{23}), (\ref%
{24}) reduces to the joint spectral measure $\mathrm{P}_{X_{1},...,X_{n}}$.
Therefore, with random variables $\pi _{X_{1}},...,\pi _{X_{n}}$ and the
measures $\{ \mu _{\rho },$ $\forall \rho \}$ in the right-hand side,
representations (\ref{28'}), (\ref{28''}) hold for all states and all finite
collections $\{X_{1},...,X_{n}\}$ of mutually commuting quantum observables.
This proves the existence point of theorem 2. Also, by proposition 1, if $%
\rho _{j}\overset{\mathfrak{R}}{\mapsto }\mu _{\rho _{j}}$, $j=1,...,m,$\
then $\sum \alpha _{j}\rho _{j}\overset{\mathfrak{R}}{\mapsto }\sum \alpha
_{j}\mu _{\rho _{j}}.$ This completes the proof.
\end{proof}

Representations (\ref{28'}), (\ref{28''}) imply.

\begin{corollary}
In the setting of theorem 2, for all states $\rho $ and all finite
collections $\{X_{1},...,X_{n}\}$ of mutually commuting quantum observables,
the representation 
\begin{equation}
\mathrm{tr}[\rho \psi (X_{1},...,X_{n})]=\dint \limits_{\Omega }\psi
(f_{X_{1}}(\omega ),...,f_{X_{n}}(\omega ))\nu _{\rho }\left( \mathrm{d}%
\omega \right)  \label{2d}
\end{equation}%
holds for all bounded Borel functions $\psi :\mathbb{R}^{n}\rightarrow 
\mathbb{R}$ and the representation 
\begin{equation}
\mathrm{tr}[\rho (X_{1}\cdot \ldots \cdot X_{n})]=\dint \limits_{\Omega
}f_{X_{1}}(\omega )\cdot \ldots \cdot f_{X_{n}}(\omega )\nu _{\rho }\left( 
\mathrm{d}\omega \right)  \label{29}
\end{equation}%
is fulfilled whenever mutually commuting quantum observables $%
X_{1},...,X_{n}\ $are bounded.
\end{corollary}

\begin{proof}
For a bounded Borel function $\psi :\mathbb{R}^{n}\rightarrow \mathbb{R},$
the quantum observable (\ref{qww}) is bounded, so that the operator $\rho
\psi (X_{1}\cdot \ldots \cdot X_{n})$ is trace class for all states $\rho .$
This and relation (\ref{28''}) imply%
\begin{eqnarray}
\mathrm{tr}[\rho \psi (X_{1},\ldots ,X_{n})] &=&\dint \limits_{\mathbb{R}%
^{n}}\psi (x_{1},...,x_{n})\text{ }\mathrm{tr}[\rho \{ \mathrm{P}_{X_{1}}(%
\mathrm{d}x_{1})\cdot \ldots \cdot \mathrm{P}_{X_{n}}(\mathrm{d}x_{n})\}] \\
&=&\dint \limits_{\Omega }\psi (f_{X_{1}}(\omega ),...,f_{X_{n}}(\omega
))\nu _{\rho }\left( \mathrm{d}\omega \right) .  \notag
\end{eqnarray}%
The proof of representation (\ref{29}) is quite similar to our proof of (\ref%
{po}). Namely, for bounded quantum observables $X_{1},...,X_{n},$ the
operator $\rho (X_{1}\cdot \ldots \cdot X_{n})$ is trace class. This and
relations (\ref{28''}), (\ref{2}) imply representation (\ref{29}). This
proves the statement.\medskip
\end{proof}

In view of expressions (\ref{vn'}), (\ref{vn}) for the von Neumann joint
probabilities, theorem 2 and corollary 2 prove that, for an arbitrary
Hilbert space $\mathcal{H},$ the Hilbert space description of all the joint
von Neumann measurements can be reproduced in terms\emph{\ }of a single
measurable space $(\Omega ,\mathcal{F}_{\Omega })$ via the set $\{ \nu
_{\rho },$ $\forall \rho \}$ of normalized real-valued measures and the set $%
\{f_{X}:\Omega \rightarrow \mathrm{sp}X,$ $f_{X}(\Omega )=\mathrm{sp}X,$ $%
X\in \mathfrak{X}_{\mathcal{H}}\}$ of random variables -- that is, via \emph{%
a} \emph{quasi hidden variable (qHV) model} \cite{12, 13}.

In this qHV model, (i) each quantum observable on $\mathcal{H}$ is modelled
on $(\Omega ,\mathcal{F}_{\Omega })$ by only one random variable
representing this quantum observable under all joint von Neumann
measurements; (ii) all the von Neumann joint probabilities are reproduced
due to the \emph{noncontextual} representations (\ref{28'}), (\ref{28''})
via nonnegative values of real-valued measures and (iii) all the quantum
product averages are reproduced due to the \emph{noncontextual
representations} (\ref{2d}) (\ref{29}) via the classical-like averages of
the corresponding expressions for the random variables. In view of the
setting of theorem 2, in this new model, the Kochen-Specker functional
conditions (\ref{03}) - (\ref{02}) do not need to hold, however, due to
representations (\ref{2d}), (\ref{29}), all these conditions are satisfied
in average -- in the sense of relations (\ref{st1}) - (\ref{st2}), discussed
in section II.

Therefore, the qHV model specified by theorem 2 and corollary 2 is \emph{%
statistically noncontextual} according to the terminology outlined in
Introduction and section II.

The specific example of such a qHV model is given in the proof of theorem 2.

For the joint von Neumann measurements performed on a quantum state $\rho ,$
theorem 2 and corollary 2 imply the following probability model.

\begin{proposition}
Let $\rho \mathcal{\ }$be a state on an arbitrary complex separable Hilbert
space. There exist a measure space $(\Omega ,\mathcal{F}_{\Omega },\nu
_{\rho })$ with a normalized real-valued measure $\nu _{\rho }$ and a set 
\begin{equation}
\{f_{X}:\Omega \rightarrow \mathrm{sp}X\mid f_{X}(\Omega )=\mathrm{sp}X,%
\text{ \ \ }X\in \mathfrak{X}_{\mathcal{H}}\}
\end{equation}%
of random variables one-to-one corresponding to the set $\mathfrak{X}_{%
\mathcal{H}}$ of all quantum observables on $\mathcal{H}$ such that all the
von Neumann joint probabilities (\ref{vn'}), (\ref{vn}) admit the
noncontextual qHV representation%
\begin{eqnarray}
\mathrm{tr}[\rho \mathrm{P}_{X_{1},...,X_{n}}(F)] &=&\nu _{\rho }\left(
f_{(X_{1},...X_{n})}^{-1}(F)\right) ,\text{ \ \ }F\in \mathcal{F}_{\mathrm{sp%
}X_{1}\times \cdots \times \mathrm{sp}X_{n}}, \\
f_{(X_{1},...,X_{n})} &:&=(f_{X_{1}},...,f_{X_{n}}),\text{ \ \ }n\in \mathbb{%
N},  \notag
\end{eqnarray}%
in particular, 
\begin{eqnarray}
\mathrm{tr}[\rho \{ \mathrm{P}_{X_{1}}(B_{1})\cdot \ldots \cdot \mathrm{P}%
_{X_{n}}(B_{n})\}] &=&\nu _{\rho }\left( \text{ }f_{X_{1}}^{-1}(B_{1})\cap
\cdots \cap f_{X_{n}}^{-1}(B_{n})\right)  \label{28} \\
&\equiv &\dint \limits_{\Omega }\chi _{f_{X_{1}}^{-1}(B_{1})}(\omega )\cdot
\ldots \cdot \chi _{f_{X_{n}}^{-1}(B_{n})}(\omega )\nu _{\rho }(\text{%
\textrm{d}}\omega ),  \notag \\
B_{i} &\in &\mathcal{B}_{\mathrm{sp}X_{i}},\text{ \ \ }i=1,...,n\in \mathbb{N%
},  \notag
\end{eqnarray}%
implying the noncontextual qHV representation%
\begin{equation}
\mathrm{tr}[\rho (X_{1}\cdot \ldots \cdot X_{n})]=\dint \limits_{\Omega
}f_{X_{1}}(\omega )\cdot \ldots \cdot f_{X_{n}}(\omega )\text{ }\nu _{\rho
}\left( \mathrm{d}\omega \right)
\end{equation}%
for the quantum product average whenever mutually commuting quantum
observables $X_{1},\ldots ,$ $X_{n},$ $n\in \mathbb{N},$ are bounded.
\end{proposition}

Recall that, in probability theory, a measurement situation (an experiment)
is generally described \cite{17} via the Kolmogorov probability model based
on the notion of a probability space \cite{7}, that is, a measure space $%
(\Omega ,\mathcal{F}_{\Omega },\tau )$ where a measure $\tau $ is a
probability one.

In view of this, proposition 2 points to the generality of the
quasi-classical probability (qHV) model introduced in Ref. 12 and
incorporating the Kolmogorov probability model as a particular case. This
new probability model is also outlined in Introduction.

The noncontextual qHV\ representation (\ref{28}) is, in particular,
fulfilled for all $N$-partite joint von Neumann measurements on an $N$%
-partite quantum state. This implies the following statement important for
quantum information applications.

\begin{proposition}
Every $N$-partite quantum state $\rho $ admits a local qHV (LqHV) model,
that is, for each state $\rho $ on a Hilbert space $\mathcal{H}_{1}\otimes
\cdots \otimes \mathcal{H}_{N},$ all the N-partite joint von Neumann
probabilities 
\begin{equation}
\mathrm{tr}[\rho \{\mathrm{P}_{X_{1}}(B_{1})\otimes \cdots \otimes \mathrm{P}%
_{X_{N}}(B_{N})\}]
\end{equation}%
admit the LqHV representation \cite{12} 
\begin{equation}
\mathrm{tr}[\rho \{\mathrm{P}_{X_{1}}(B_{1})\otimes \cdots \otimes \mathrm{P}%
_{X_{N}}(B_{N})\}]=\dint\limits_{\Omega }P_{X_{1}}(B_{1};\omega )\cdot
\ldots \cdot P_{X_{N}}(B_{N};\omega )\nu _{\rho }(\text{\textrm{d}}\omega )
\label{zz}
\end{equation}%
in terms of a single measure space $(\Omega ,\mathcal{F}_{\Omega },\nu
_{\rho }),$ with a normalized real-valued measure $\nu _{\rho },$ and
conditional probability distributions $P_{X_{n}}(\cdot $ $;\omega ),$ $%
n=1,...,N,$ each depending only on the corresponding quantum observable $%
X_{n}$ on a Hilbert space $\mathcal{H}_{n}$ at $n$-th site.
\end{proposition}

\begin{proof}
For a state $\rho $ and quantum observables 
\begin{eqnarray}
&&X_{1}\otimes \mathbb{I}_{\mathcal{H}_{2}}\otimes \cdots \otimes \mathbb{I}%
_{\mathcal{H}_{N}}, \\
&&...,  \notag \\
&&\mathbb{I}_{\mathcal{H}_{1}}\otimes \cdots \otimes \mathbb{I}_{\mathcal{H}%
_{N-1}}\otimes X_{\mathcal{H}_{N}},  \notag
\end{eqnarray}%
on $\mathcal{H}_{1}\otimes \cdots \otimes \mathcal{H}_{N},$ representation (%
\ref{28}) has the LqHV form \cite{12}%
\begin{eqnarray}
\mathrm{tr}[\rho \{ \mathrm{P}_{X_{1}}(B_{1})\otimes \ldots \otimes \mathrm{P%
}_{X_{N}}(B_{N})\}] &=&\dint \limits_{\Omega }\chi
_{f_{X_{1}}^{-1}(B_{1})}(\omega )\cdot \ldots \cdot \chi
_{f_{X_{N}}^{-1}(B_{N})}(\omega )\nu _{\rho }(\text{\textrm{d}}\omega ), \\
B_{1} &\in &\mathcal{B}_{\mathrm{sp}X_{1}},...,B_{N}\in \mathcal{B}_{\mathrm{%
sp}X_{N}}.  \notag
\end{eqnarray}%
This proves the statement.\medskip
\end{proof}

We stress that, in theorem 2, a mapping $\Phi $ does not need to satisfy for
each Borel function $\varphi :\mathbb{R}\rightarrow \mathbb{R}$ and each
quantum observable $X$ the functional conditions (\ref{03}) - (\ref{02})
required by Kochen and Specker in Ref. 5. Moreover, in view of the
Kochen-Specker result \cite{5}, for a mapping $\Phi $ in theorem 2, these
conditions cannot be fulfilled whenever $\dim \mathcal{H}\geq 3.$

Accordingly, in a statistically noncontextual qHV\ model specified by
theorem 2 and corollary 2, the Kochen-Specker conditions (\ref{03}) - (\ref%
{02}) do not need to hold if $\dim \mathcal{H}=2$ and cannot be fulfilled
for $\dim \mathcal{H}\geq 3$. \emph{Nevertheless,} for each Hilbert space%
\textbf{\ }$\mathcal{H},$\textbf{\ }this new qHV\ model reproduces all the
properties of the von Neumann joint probabilities and product expectations
via "noncontextual" random variables. Also, in this new qHV model, all the
physically motivated average relations (\ref{st1}) - (\ref{st2}), argued in
question 2 of section II, follow explicitly from the noncontextual
representations (\ref{2d}), (\ref{29}).

\subsection{A context-invariant qHV model}

It is generally argued that if a mapping $\Psi $ from a set of random
variables on $(\Omega ,$ $\mathcal{F}_{\Omega })$ onto the set of all
quantum observable on $\mathcal{H}$\emph{\ }is non-injective,\emph{\ }then
the corresponding model reproducing in terms of random variables the
statistical properties of all quantum observables on $\mathcal{H}$ needs to
be \emph{contextual} in the sense that, under a joint von Neumann
measurement of mutually commuting quantum observables $X_{1},...,X_{n},$
each observable $X_{i}$ is modelled by a random variable specific for a
context of this joint measurement.

In what follows, we prove (theorem 3) that, for all joint von Neumann
measurements, the existence of a statistically noncontextual qHV model,
specified by theorem 2 and corollary 2, implies the existence of \emph{a
context-invariant qHV model} -- a model of a completely new type described,
in general, in Introduction and question 3 formulated in section II.

In this new qHV model, the functional condition (\ref{05}), required in
quantum foundations, is fulfilled for each Hilbert space, nevertheless, this
model is not contextual.

Consider first the following property proved in appendix C.

\begin{lemma}
Let $\Phi :\mathfrak{X}_{\mathcal{H}}\rightarrow \mathfrak{F}_{(\Omega ,%
\mathcal{F}_{\Omega })},$ $f_{X}:=\Phi (X)$ and $\{ \nu _{\rho },\forall
\rho \}$ be, correspondingly, the mapping, random variables and the measures
specified in theorem 2. Then%
\begin{eqnarray}
&&\nu _{\rho }\mathbb{(}f_{\varphi (X)}^{-1}(B)\cap
f_{Y_{1}}^{-1}(B_{1})\cap \cdots \cap f_{Y_{m}}^{-1}(B_{m}))  \label{pp} \\
&=&\nu _{\rho }\mathbb{((}\varphi \circ f_{X})^{-1}(B)\cap
f_{Y_{1}}^{-1}(B_{1})\cap \cdots \cap f_{Y_{m}}^{-1}(B_{m})),  \notag \\
B &\in &\mathcal{B}_{\mathrm{sp}\varphi (X)},\text{ \ \ }B_{i}\in \mathcal{B}%
_{\mathrm{sp}Y_{i}},\text{ \ \ \ }i=1,...,m\in \mathbb{N},  \notag
\end{eqnarray}%
for all states $\rho ,$ all Borel functions $\varphi :\mathbb{R}\rightarrow 
\mathbb{R}$ \ and all finite collections $\{X,Y_{1},...,Y_{m}\}$ of mutually
commuting quantum observables on $\mathcal{H}.$
\end{lemma}

From relations (\ref{vn}), (\ref{28''}) and property (\ref{pp}) it follows
that if $f_{X},$ $f_{\varphi (X)}\in \Phi (\mathfrak{X}_{\mathcal{H}})$ are
random variables specified in theorem 2, then the random variables $%
f_{\varphi (X)}$ and $\varphi \circ f_{X}$ on $(\Omega ,\mathcal{F}_{\Omega
})$ equivalently represent the quantum observable $\varphi (X)$ \emph{under
all joint von Neumann measurements }regardless of their measurement
contexts. However,\emph{\ }for a Borel function $\varphi :\mathbb{R}%
\rightarrow \mathbb{R}$ and a quantum observable $X$, the random variable $%
\varphi \circ f_{X}$ \emph{does not need} to belong to the set $\Phi (%
\mathfrak{X}_{\mathcal{H}})$ and, moreover, to coincide with the random
variable $f_{\varphi (X)}$.

The following theorem (proved in appendix D) answers positively to question
3 put in section II.

\begin{theorem}
Let $\mathcal{H}$ be an arbitrary complex separable Hilbert space. There
exist:\newline
(i) a measurable space $(\Omega ,$ $\mathcal{F}_{\Omega });$\newline
(ii) a mapping $\Psi :\mathfrak{F}_{\mathfrak{X}_{\mathcal{H}}}\rightarrow 
\mathfrak{X}_{\mathcal{H}}$ from a set $\mathfrak{F}_{\mathfrak{X}_{\mathcal{%
H}}}$ of random variables $g\ $on $(\Omega ,$ $\mathcal{F}_{\Omega })$ onto
the set $\mathfrak{X}_{\mathcal{H}}$ of all quantum observables on $\mathcal{%
H}$, with the spectral correspondence rule $g(\Omega )=\mathrm{sp}X$ for
each \cite{uu} $g\in \Psi ^{-1}(\{X\})$ and the functional relations 
\begin{eqnarray}
\varphi \circ g &\in &\mathfrak{F}_{\mathfrak{X}_{\mathcal{H}}},\text{ \ \ \ 
}\Psi (\varphi \circ g)=\varphi \circ \Psi (g)=\varphi \circ X,  \label{ff}
\\
\forall g &\in &\Psi ^{-1}(\{X\}),  \notag
\end{eqnarray}%
for all Borel functions $\varphi :\mathbb{R}\rightarrow \mathbb{R}$ and all
quantum observables $X;$\newline
such that, to each quantum state $\rho $ on $\mathcal{H},$ there corresponds 
$(\rho \overset{\mathfrak{R}}{\mapsto }\nu _{\rho })$ a unique normalized
real-valued measure $\nu _{\rho }$ on $(\Omega ,\mathcal{F}_{\Omega })$
satisfying the context-invariant representation%
\begin{eqnarray}
\mathrm{tr}[\rho \{ \mathrm{P}_{X_{1}}(B_{1})\cdot \ldots \cdot \mathrm{P}%
_{X_{n}}(B_{n})\}] &=&\nu _{\rho }\left( g_{1}^{-1}(B_{1})\cap \cdots \cap
g_{n}^{-1}(B_{n})\right) ,  \label{kl} \\
\forall g_{i} &\in &\Psi ^{-1}(\{X_{i}\}),  \notag
\end{eqnarray}%
for all sets $B_{i}\in \mathcal{B}_{\mathrm{sp}X_{i}},$\ $i=1,...,n,$ and
all collections $\{X_{1},...,X_{n}\} \subset \mathfrak{X}_{\mathcal{H}},$ $%
n\in \mathbb{N},$ of mutually commuting quantum observables on $\mathcal{H}.$
If $\rho _{j}\overset{\mathfrak{R}}{\mapsto }\nu _{\rho _{j}},$ $%
j=1,...,m<\infty ,$ then $\sum \alpha _{j}\rho _{j}\overset{\mathfrak{R}}{%
\mapsto }\sum \alpha _{j}\nu _{\rho _{j}},$ for all $\alpha _{j}>0,$ \ $\sum
\alpha _{j}=1.$
\end{theorem}

Theorem 3 implies the following statement proved in appendix D.

\begin{corollary}
In the setting of theorem 3, for all states $\rho $ and all finite
collections $\{X_{1},...,X_{n}\} \subset \mathfrak{X}_{\mathcal{H}}$ of
mutually commuting quantum observables on $\mathcal{H},$ the
context-invariant representation 
\begin{eqnarray}
\mathrm{tr}[\rho \psi (X_{1},...,X_{n})] &=&\dint \limits_{\Omega }\psi
(g_{1}(\omega ),...,g_{n}(\omega ))\nu _{\rho }\left( \mathrm{d}\omega
\right) ,  \label{gg'} \\
\forall g_{i} &\in &\Psi ^{-1}(\{X_{i}\}),\text{ \ \ }i=1,...,n,  \notag
\end{eqnarray}%
holds for all bounded Borel functions $\psi :\mathbb{R}^{n}\rightarrow 
\mathbb{R}$ and the context-invariant representation 
\begin{eqnarray}
\mathrm{tr}[\rho (X_{1}\cdot \ldots \cdot X_{n})] &=&\dint \limits_{\Omega
}g_{1}(\omega )\cdot \ldots \cdot g_{n}(\omega )\nu _{\rho }\left( \mathrm{d}%
\omega \right) ,  \label{gg} \\
\forall g_{i} &\in &\Psi ^{-1}(\{X_{i}\}),\text{ \ \ }i=1,...,n,  \notag
\end{eqnarray}%
is fulfilled whenever mutually commuting quantum observables $X_{1},\ldots
,X_{n}$ are bounded.
\end{corollary}

\begin{remark}
Representations (\ref{kl}) - (\ref{gg}) are context-invariant in the sense
that, regardless of a context of a joint von Neumann measurement, into the
right-hand sides of these representations, each of random variables $%
g_{i}\in \Psi ^{-1}(\{X_{i}\})$, modelling on $(\Omega ,\mathcal{F}_{\Omega
})$ a quantum observable $X_{i}$, can be equivalently substituted\textbf{.}
\end{remark}

From theorem 3 and corollary 3 it follows that, for each Hilbert space $%
\mathcal{H},$ the Hilbert space description of all the joint von Neumann
measurements (equivalently, all the quantum observables on $\mathcal{H}$)
can be reproduced via a qHV model, where a quantum observable $X$ on $%
\mathcal{H}$ can be represented on $(\Omega ,\mathcal{F}_{\Omega })$ by a
variety of random variables, but, \emph{due to the} \emph{context-invariant
representations} (\ref{kl}) - (\ref{gg}), each of these random variables 
\emph{equivalently} models $X$ under all joint von Neumann measurements,
regardless of their measurement contexts. In this new qHV model, the
functional condition (\ref{05}) required in quantum foundations is fulfilled
for each Hilbert space and constitutes condition (\ref{ff}) in theorem 3.
All the relations

Therefore, the qHV model specified by theorem 3 and corollary 3 is \emph{%
context-invariant}\textbf{\emph{\ }}according to the terminology outlined in
Introduction and section II.

The specific example of such a qHV\ model is given in the proof of theorem 3.

In view of relation (\ref{ff}) and the Kochen-Specker result \cite{5}, for $%
\dim \mathcal{H}\geq 3,$ a mapping $\Psi $ in theorem 3 cannot be injective.
For $\dim \mathcal{H}=2,$ this mapping does not need to be injective due to
the setting of theorem 3.

Accordingly, a context-invariant qHV model specified by theorem 3 cannot be
noncontextual for $\dim \mathcal{H}\geq 3$ and does not need to be
noncontextual if $\dim \mathcal{H}=2.$\emph{\ Nevertheless}, in this new
model (in contrast to a contextual one), all the von Neumann joint
probabilities and product expectations are reproduced via \emph{the
context-invariant representations} (\ref{kl}) - (\ref{gg}).

Thus, in contrast to the wide-spread opinion, a model reproducing the
statistical properties of all quantum observables via random variables \emph{%
does not need to be contextual }whenever\emph{, }in this model, a mapping $%
\Psi $ from a set of random variables onto the set of all quantum
observables is non-injective.\textbf{\ }

\section{Conclusions}

In the present paper, we have introduced the two new quasi hidden variable
(qHV) models correctly reproducing \emph{in measure theory terms} the
Hilbert space description of all joint von Neumann measurements for an
arbitrary Hilbert space.\emph{\ }These new models answer positively to
either of three questions formulated in section II and are important for
both -- quantum foundations and quantum applications.

For this aim, we had first to generalize (lemma 3) some items of the
Kolmogorov extension theorem \cite{7, 17} to the case of consistent
operator-valued measures. This generalization allowed us to express (theorem
1) all symmetrized finite products of the spectral measures of quantum
observables via a unique self-adjoint operator-valued measure and the
specific random variables defined on some specially constructed measurable
space.

Based on these new mathematical results, we further analyzed modelling of
all the von Neumann joint probabilities and product expectations in qHV \cite%
{12, 13} terms. We have proved that, for each Hilbert space, the Hilbert
space description of all joint von Neumann measurements can be reproduced
via either of the two new quasi hidden variable (qHV) models, \emph{%
statistically} \emph{noncontextual and context-invariant.}

In both of these new qHV\ models, all the von Neumann joint probabilities
are represented (theorems 2, 3) via nonnegative values of real-valued
measures and all the quantum product expectations -- via the qHV average
(corollaries 2, 3) of the product of the corresponding random variables.

For $\dim \mathcal{H\geq }4$, the HV\ versions of these new models cannot
exist (see section II).

\emph{In a statistically noncontextual qHV\ mode}$\emph{l}$ specified by
theorem 2 and corollary 2, each quantum observable $X$ on a Hilbert space $%
\mathcal{H}$ is represented on a measurable space $(\Omega ,\mathcal{F}%
_{\Omega })$ \emph{by only one} random variable $f_{X}$, modelling this
observable $X$ under all joint von Neumann measurements. In this new model,
the Kochen-Specker \cite{5} functional conditions (\ref{03}) - (\ref{02}) on
a mapping $X\mapsto f_{X}$ do not need to hold and, moreover, in view of the
Kochen-Specker result \cite{5}, cannot hold whenever $\dim \mathcal{H}\geq 3$%
.

\emph{Nevertheless,} this new qHV model reproduces all the von Neumann joint
probabilities and product expectations via the \emph{noncontextual} \emph{%
representations} (\ref{28'}) - (\ref{29}). Moreover, in this model, the
Kochen-Specker functional conditions are satisfied in average -- in the
sense of relations (\ref{st1}) - (\ref{st2}) discussed in section II.

The specific example of a statistically noncontextual qHV\ model for all the
joint von Neumann measurements is given in the proof of theorem 2.

The proved existence of a statistically noncontextual qHV model, in
particular, implies (proposition 3) that \emph{every }$N$\emph{-partite
quantum state admits a local qHV (LqHV) model} introduced in Ref. 10.

This new result is particularly important for quantum applications since it
is specifically the qHV approach that allowed us to construct \cite{12} the
quantum analogs of Bell-type inequalities \cite{11} and to find \cite{12}
the new exact upper bounds on violation of a Bell-type inequality by a $N$%
-partite quantum state -- the problem which has been intensively discussed
in the literature since the publication \cite{tsirelson} of Tsirelson and
which is now important for a variety of quantum information processing tasks.

\emph{In a context-invariant qHV\ model} specified by\ theorem 3 and
corollary 3, a quantum observable $X$\ on a Hilbert space $\mathcal{H}$%
\textbf{\ }can be modelled on a measurable space\textbf{\ }$(\Omega ,%
\mathcal{F}_{\Omega })$\textbf{\ }by a variety of random variables
satisfying the functional condition (\ref{ff}) (equivalently, (\ref{05}))
required in quantum foundations, but each of these random variables
equivalently represents $X$\ under all joint von Neumann measurements,
independently of their measurement contexts.\textbf{\ }

For\textbf{\ }$\dim \mathcal{H}\geq 3,$ a context-invariant qHV\ model
cannot be noncontextual. For\textbf{\ }$\dim \mathcal{H}=2,$\textbf{\ }this
qHV\ model does not need to be noncontextual. \emph{Nevertheless}, in
contrast to a contextual HV model, this new qHV model reproduces all the von
Neumann joint probabilities and product expectations\textbf{\ }via \emph{the
context-invariant representations} (\ref{kl}) - (\ref{gg}) where, regardless
of a context of a joint von Neumann measurement, each of random variables
modelling a quantum observable $X_{i}$ can be equivalently substituted into
the right-hand sides.

The specific example of a context-invariant qHV model for all the joint von
Neumann measurements is presented in the proof of theorem 3.

The\ proved\ existence for each Hilbert space of\emph{\ a context-invariant
qHV\ model\ negates }$\emph{t}$\emph{he general opinion} that, in terms of
random variables satisfying the functional condition (\ref{05}) required in
quantum foundations the Hilbert space description of all the joint von
Neumann measurements (equivalently, all the quantum observables) for $\dim 
\mathcal{H}\geq 3$ can be reproduced only contextually.

We stress that it is specifically the real-valued measures and the random
variables introduced in proposition 1 that allowed us to prove the existence
for all the joint von Neumann measurements of the noncontextual
representations (\ref{28'}) - (\ref{29}) and the context-invariant
representations (\ref{kl}) - (\ref{gg}). Accordingly, the appearance of
these real-valued measures in a model is due to our accounting of the
structure of projection-valued measures describing joint von Neumann
measurements.

The results of the present paper also underline the generality of the
quasi-classical probability model proposed in Ref. 12.

\begin{acknowledgement}
I am very grateful to Professor V. N. Samarov for his helpful and
encouraging comments. The discussions with Professor M. E. Shirokov are much
appreciated.
\end{acknowledgement}

\appendix

\section{proof of lemma 1}

Relation (\ref{13}) follows explicitly from the symmetrized form of the
right-hand side in (\ref{9_}). In order to prove (\ref{14}), let us first
take a collection $\{X_{1},...,X_{n}\}$ of bounded quantum observables with
discrete spectrums. In this case, the measure $\mathcal{P}%
_{(X_{1},...,X_{_{n}})}$ is given by representation (\ref{9}) and taking
into the account that $\mathrm{P}_{X_{i}}(\mathrm{sp}X_{i})=\mathbb{I}_{%
\mathcal{H}},$ we have:%
\begin{eqnarray}
&&\mathcal{P}_{(X_{1},...,X_{_{n}})}\left( \{(x_{1},...,x_{n})\in \mathrm{sp}%
X_{1}\times \cdots \times \mathrm{sp}X_{n}\mid (x_{i_{1}},...,x_{i_{k}})\in
F\}\right)  \TCItag{A1}  \label{A1} \\
&=&\frac{1}{n!}\dsum\limits_{(x_{i_{1}},...,x_{i_{k}})\in F}\text{ }\left\{ 
\text{ }\mathrm{P}_{X_{1}}(\{x_{1}\})\cdot \ldots \cdot \mathrm{P}%
_{X_{n}}(\{x_{n}\})\right\} _{\mathrm{sym}}  \notag \\
&=&\frac{1}{k!}\dsum\limits_{(x_{i_{1}},...,x_{i_{k}})\in F}\left\{ \mathrm{P%
}_{X_{i_{1}}}(\{x_{i_{1}}\})\cdot \ldots \cdot \mathrm{P}_{X_{i_{k}}}(%
\{x_{i_{k}}\})\right\} _{\mathrm{sym}}  \notag \\
&=&\mathcal{P}_{(X_{i_{1}},...,X_{i_{k}})}\left( F\right) .  \notag
\end{eqnarray}

For an arbitrary collection $\{X_{1},...,X_{n}\}$ $\subset \mathfrak{X}_{%
\mathcal{H}}$ of quantum observables on $\mathcal{H},$ let $\mathcal{E}$ be
the set of all rectangles $E:=$ $B_{i_{1}}\times \cdots \times B_{i_{k}}$
with measurable sides $B_{i}\in \mathcal{B}_{\mathrm{sp}X_{i}}.$ Since the
product algebra $\mathcal{F}_{\mathrm{sp}X_{i_{1}}\times \cdots \times 
\mathrm{sp}X_{i_{k}}}$ consists of all finite unions of mutually disjoint
rectangles from $\mathcal{E}$, every set $F\in \mathcal{F}_{\mathrm{sp}%
X_{i_{1}}\times \cdots \times \mathrm{sp}X_{i_{k}}}$ admits a finite
decomposition 
\begin{equation}
F=\underset{m=1,...,M}{\cup }E_{m},\text{ \ \ }E_{m_{1}}\cap
E_{m_{2}}=\varnothing ,\ \ E_{m}\in \mathcal{E},\text{ \ }M<\infty . 
\tag{A2}  \label{A2}
\end{equation}%
Taking into the account that $\mathcal{P}_{(X_{1},...,X_{_{n}})}$ and $%
\mathcal{P}_{(X_{i_{1}},...,X_{i_{k}})}$ are finitely additive measures and
also, relations (\ref{A2}), (\ref{9_}) and $\mathrm{P}_{X_{i}}(\mathrm{sp}%
X_{i})=\mathbb{I}_{\mathcal{H}},$ we have:%
\begin{eqnarray}
&&\mathcal{P}_{(X_{1},...,X_{_{n}})}\left( \{(x_{1},...,x_{n})\in \mathrm{sp}%
X_{1}\times \cdots \times \mathrm{sp}X_{n}\mid (x_{i_{1}},...,x_{i_{k}})\in
F\} \right)  \TCItag{A3}  \label{A3} \\
&=&\sum_{m=1,...,M}\mathcal{P}_{(X_{1},...,X_{_{n}})}\left(
\{(x_{1},...,x_{n})\in \mathrm{sp}X_{1}\times \cdots \times \mathrm{sp}%
X_{n}\mid (x_{i_{1}},...,x_{i_{k}})\in E_{m}\} \right)  \notag \\
&=&\sum_{m=1,...,M}\mathcal{P}_{(X_{i_{1}},...,X_{i_{k}})}\left( E_{m}\right)
\notag \\
&=&\mathcal{P}_{(X_{i_{1}},...,X_{i_{k}})}\left( F\right) .  \notag
\end{eqnarray}%
This proves lemma 1.

\section{a generalization of the Kolmogorov extension theorem}

In this appendix, we generalize (lemma 3) to the case of consistent
operator-valued measures some items of the Kolmogorov consistency theorem 
\cite{7, 17} for a family of consistent probability measures (\ref{15}).

For an uncountable index set $T,$ consider a family $\{(\Lambda _{t},%
\mathcal{F}_{\Lambda _{t}}),$ $t\in T\}$ of measurable spaces, where each $%
\Lambda _{t}$ is a non-empty set and $\mathcal{F}_{\Lambda _{t}}$ is an
algebra of subsets of $\Lambda _{t}$.

Let $\mathcal{F}_{\Lambda _{t_{1}}\times \cdots \times \Lambda _{t_{n}}}$ be
the product algebra on $\Lambda _{t_{1}}\times \cdots \times \Lambda
_{t_{n}} $, that is, the algebra generated by all rectangles $F_{1}\times
\cdots \times F_{n}\subseteq \Lambda _{t_{1}}\times \cdots \times \Lambda
_{t_{n}}$ with measurable sides $F_{k}\in \mathcal{F}_{\Lambda _{t_{k}}}.$

Denote by $\widetilde{\Lambda }:=\dprod \limits_{t\in T}\Lambda _{t}$ the
Cartesian product \cite{dunford} of all sets $\Lambda _{t},$ $t\in T$, that
is, $\widetilde{\Lambda }$ is the collection of all functions $\lambda
:T\rightarrow $ $\cup _{t\in T}\Lambda _{t}$ with values $\lambda
(t):=\lambda _{t}\in \Lambda _{t}.$

The set of all cylindrical subsets of $\widetilde{\Lambda }$ of the form 
\begin{equation}
\mathcal{J}_{(t_{1},...,t_{n})}(F):=\{\lambda \in \widetilde{\Lambda }\mid
(\lambda _{t_{1}},...,\lambda _{t_{n}})\in F\}\text{, \ \ \ \ }F\in \mathcal{%
F}_{\Lambda _{t_{1}}\times \cdots \times \Lambda _{t_{n}}},  \label{1_}
\end{equation}%
where $\{t_{1},...,t_{n}\}\subset T,$\ $n\in \mathbb{N},$ constitutes \cite%
{dunford, 17} an algebra on $\widetilde{\Lambda }$ that we further denote by 
$\mathcal{A}_{\widetilde{\Lambda }}.$ Since $\mathcal{J}%
_{(t_{1},...,t_{n})}(F)\equiv \pi _{(t_{1},...,t_{n})}^{-1}(F),$ where the
function%
\begin{equation}
\pi _{(t_{1},...,t_{n})}:\widetilde{\Lambda }\rightarrow \Lambda
_{t_{1}}\times \cdots \times \Lambda _{t_{n}}
\end{equation}%
is the canonical projection 
\begin{equation}
\pi _{(t_{1},...,t_{n})}(\lambda ):=(\pi _{t_{1}}(\lambda ),...,\pi
_{t_{n}}(\lambda )),\text{ \ \ \ \ \ }\pi _{t}(\lambda ):=\lambda _{t},
\label{2_}
\end{equation}%
we have 
\begin{equation}
\mathcal{A}_{\widetilde{\Lambda }}=\{\pi
_{(t_{1},...,t_{n})}^{-1}(F)\subseteq \widetilde{\Lambda }\text{ }\mid \text{
}F\in \mathcal{F}_{\Lambda _{t_{1}}\times \cdots \times \Lambda _{t_{n}}},%
\text{\ \ }\{t_{1},...,t_{n}\}\subset T,\text{\ \ }n\in \mathbb{N}\}.
\label{3_}
\end{equation}%
\smallskip 

Introduce a family 
\begin{equation}
\{\mathfrak{M}_{(t_{1},...,t_{n})}:\mathcal{F}_{\Lambda _{t_{1}}\times
\cdots \times \Lambda _{t_{n}}}\rightarrow \mathcal{L}_{\mathcal{H}%
}^{(s)}\mid \mathfrak{M}_{(t_{1},...,t_{n})}(\Lambda _{t_{1}}\times \cdots
\times \Lambda _{t_{n}})=\mathbb{I}_{\mathcal{H}},\text{ \ }%
\{t_{1},...,t_{n}\}\subset T,\text{ \ }n\in \mathbb{N}\}  \label{4_}
\end{equation}%
of normalized finitely additive $\mathcal{L}_{\mathcal{H}}^{(s)}$-valued
measures $\mathfrak{M}_{(t_{1},...,t_{n})}$, each specified by a finite
collection $\{t_{1},...,t_{n}\}\subset T$ \ of indices and having values
that are self-adjoint bounded linear operators on $\mathcal{H}.$

Let, for each finite index collection $\{t_{1},...,t_{n}\}\subset T,$ these
measures satisfy the consistency condition 
\begin{eqnarray}
\mathfrak{M}_{(t_{1},...,t_{n})}(F_{1}\times \cdots \times F_{n}) &=&%
\mathfrak{M}_{(t_{i_{1}},...,t_{i_{_{n}}})}(F_{i_{1}}\times \cdots \times
F_{i_{n}}),  \label{5_} \\
F_{i} &\in &\mathcal{F}_{\Lambda _{t_{i}}},\text{ \ \ }i=1,...,n,  \notag
\end{eqnarray}%
for all permutations $\binom{1,,...,n}{i_{1},...,i_{n}}$ of indexes and the
consistency condition 
\begin{eqnarray}
&&\mathfrak{M}_{(t_{1},...,t_{_{n}})}\left( \{(\lambda _{1},...,\lambda
_{n})\in \Lambda _{t_{1}}\times \cdots \times \Lambda _{t_{n}}\mid (\lambda
_{i_{1}},...,\lambda _{i_{k}})\in F\}\right)  \label{6__} \\
&=&\mathfrak{M}_{(t_{i_{1}},...,t_{i_{k}})}\left( F\right) ,\text{ \ \ \ \ \ 
}F\in \mathcal{F}_{\Lambda _{t_{i_{1}}}\times \cdots \times \Lambda
_{t_{i_{_{k}}}}},  \notag
\end{eqnarray}%
for each $\{t_{i_{1}},...,t_{i_{k}}\}\subseteq
\{t_{1},...,t_{n}\}.\smallskip $

\begin{lemma}
Let $\mathcal{H}$ be a complex separable Hilbert space. For a family (\ref%
{4_}) of normalized finitely additive $\mathcal{L}_{\mathcal{H}}^{(s)}$%
-valued measures $\mathfrak{M}_{(t_{1},...,t_{n})}$ satisfying the
consistency conditions (\ref{5_}) and (\ref{6__}), there exists a unique
normalized finitely additive $\mathcal{L}_{\mathcal{H}}^{(s)}$-valued
measure 
\begin{equation}
\mathbb{M}:\mathcal{A}_{\widetilde{\Lambda }}\rightarrow \mathcal{L}_{%
\mathcal{H}}^{(s)},\ \ \ \ \mathbb{M(}\widetilde{\Lambda })=\mathbb{I}_{%
\mathcal{H}},  \label{7_}
\end{equation}%
on $(\widetilde{\Lambda },\mathcal{A}_{\widetilde{\Lambda }})$ such that 
\begin{equation}
\mathbb{M}\left( \pi _{(t_{1},...,t_{n})}^{-1}(F)\right) =\mathfrak{M}%
_{(t_{1},...,t_{_{n}})}(F)  \label{8_}
\end{equation}%
for all sets $F\in \mathcal{F}_{\Lambda _{t_{1}}\times \cdots \times \Lambda
_{t_{n}}}$ and an arbitrary index collection $\{t_{1},...,t_{n}\}\subset T,$ 
$n\in \mathbb{N}.$\medskip
\end{lemma}

\begin{proof}
Our proof of lemma 3 is quite similar to the proof \cite{7, 17} of the
corresponding items in the Kolmogorov extension theorem for consistent
probability measures (\ref{15}). Let $\mathcal{A}_{\widetilde{\Lambda }}$ be
algebra (\ref{3_}) on $\widetilde{\Lambda }$. Suppose that a set $A\in 
\mathcal{A}_{\widetilde{\Lambda }}$ admits representation $A=\pi
_{(t_{1},...,t_{n})}^{-1}(F),$ where $\{t_{1},...,t_{n}\}\subset T$ and $%
F\in \mathcal{F}_{\Lambda _{t_{1}}\times \cdots \times \Lambda _{t_{n}}}$,
and take%
\begin{equation}
\mathbb{M}(A):=\mathfrak{M}_{(t_{1},...,t_{_{n}})}(F).  \label{123}
\end{equation}%
In order to show that relation (\ref{123}) defines correctly a set function $%
\mathbb{M}$ on $\mathcal{A}_{\widetilde{\Lambda }},$ we must prove that this
relation implies a unique value of $\mathbb{M}$ on a set $A$ even if this
set $A$ admits two different representations, say: 
\begin{eqnarray}
A &=&\pi _{(t_{i_{1}},...,t_{i_{k}})}^{-1}(F)\equiv \{\lambda \in \widetilde{%
\Lambda }\mid (\lambda _{t_{i_{1}}},...,\lambda _{t_{i_{k}}})\in F\},
\label{10_} \\
A &=&\pi _{(t_{j_{1}},...,t_{j_{_{m}}})}^{-1}(F^{\prime })\equiv \{\lambda
\in \widetilde{\Lambda }\mid (\lambda _{t_{j_{1}}},...,\lambda
_{t_{j_{m}}})\in F^{\prime }\}  \notag
\end{eqnarray}%
for some sets $F\in \mathcal{F}_{\Lambda _{t_{i_{1}}}\times \cdots \times
\Lambda _{t_{i_{k}}}}$ and $F^{\prime }\in \mathcal{F}_{\Lambda
_{t_{j_{1}}}\times \cdots \times \Lambda _{t_{j_{m}}}}$ and some index
collections $\{t_{i_{1}},...,t_{i_{k}}\},$ $\{t_{j_{1}},...,t_{j_{_{m}}}\}%
\subset T$ .

Denote 
\begin{equation}
\{t_{i_{1}},...,t_{i_{k}}\}\cup
\{t_{j_{1}},...,t_{j_{_{m}}}\}=\{t_{1},...,t_{_{n}}\}.  \label{11_}
\end{equation}%
From representations (\ref{10_}) it follows that sets $F$ and $F^{\prime }$
are such that, for a point $(\lambda _{1},...,\lambda _{n})$ in the set $%
\Lambda _{t_{1}}\times \cdots \times \Lambda _{t_{n}},$ the condition $%
(\lambda _{i_{1}},...,\lambda _{i_{k}})\in F$ and the condition $(\lambda
_{j_{1}},...,\lambda _{j_{m}})\in F^{\prime }$ are equivalent, that is:%
\begin{eqnarray}
&&\left\{ \text{ }(\lambda _{1},...,\lambda _{n})\in \Lambda _{t_{1}}\times
\cdots \times \Lambda _{t_{n}}\mid (\lambda _{i_{1}},...,\lambda
_{i_{k}})\in F\right\}   \label{12_} \\
&=&\left\{ \text{ }(\lambda _{1},...,\lambda _{n})\in \Lambda _{t_{1}}\times
\cdots \times \Lambda _{t_{n}}\mid (\lambda _{j_{1}},...,\lambda
_{j_{m}})\in F^{\prime }\right\} .  \notag
\end{eqnarray}%
Due to relations (\ref{123}) - (\ref{12_}) and the consistency conditions (%
\ref{5_}) and (\ref{6__}), we have:%
\begin{eqnarray}
\mathbb{M}(\pi _{(t_{i_{1}},...,t_{i_{k}})}^{-1}(F)) &=&\mathfrak{M}%
_{(t_{i_{1}},...,t_{i_{k}})}(F)  \label{13_} \\
&=&\mathfrak{M}_{(t_{1},...,t_{_{n}})}\left( \left\{ (\lambda
_{1},...,\lambda _{n})\in \Lambda _{t_{1}}\times \cdots \times \Lambda
_{t_{n}}\mid (\lambda _{i_{1}},...,\lambda _{i_{k}})\in F\right\} \right)  
\notag \\
&=&\mathfrak{M}_{(t_{1},...,t_{_{n}})}\left( \left\{ (\lambda
_{1},...,\lambda _{n})\in \Lambda _{t_{1}}\times \cdots \times \Lambda
_{t_{n}}\mid (\lambda _{j_{1}},...,\lambda _{j_{m}})\in F^{\prime }\right\}
\right)   \notag \\
&=&\mathfrak{M}_{(t_{j_{1}},...,t_{j_{m}})}(F^{\prime })  \notag \\
&=&\mathbb{M(}\pi _{(t_{j_{1}},...,t_{j_{_{m}}})}^{-1}(F^{\prime })).  \notag
\end{eqnarray}

Thus, relation (\ref{123}) defines a unique set function $\mathbb{M}:%
\mathcal{A}_{\Lambda }\rightarrow \mathcal{L}_{\mathcal{H}}$ satisfying
condition (\ref{8_}). Since $\widetilde{\Lambda }=\pi
_{(t_{1},...,t_{n})}^{-1}(\Lambda _{t_{1}}\times \cdots \times \Lambda
_{t_{n}})$ and $\mathfrak{M}_{(t_{1},...,t_{n})}(\Lambda _{t_{1}}\times
\cdots \times \Lambda _{t_{n}})=\mathbb{I}_{\mathcal{H}}$, from (\ref{123})
it follows that this set function $\mathbb{M}$ is normalized, that is, $%
\mathbb{M}(\widetilde{\Lambda })=\mathbb{I}_{\mathcal{H}}.$

In order to prove that the normalized set function $\mathbb{M}:\mathcal{A}_{%
\widetilde{\Lambda }}\rightarrow \mathcal{L}_{\mathcal{H}}$ is additive, let
us consider in the algebra $\mathcal{A}_{\Lambda }$ two disjoint sets 
\begin{equation}
A_{1}=\pi _{(t_{i_{1}},...,t_{i_{k}})}^{-1}(F_{1}),\text{ \ \ }A_{2}=\pi
_{(t_{j_{1}},...,t_{j_{m}})}^{-1}(F_{2}),\text{ \ \ }A_{1}\cap
A_{2}=\varnothing ,  \label{14_}
\end{equation}%
specified by some index collections $\{t_{i_{1}},...,t_{i_{k}}\},$ $%
\{t_{j_{1}},...,t_{j_{m}}\}\subseteq $ $\{t_{1},...,t_{n}\}\subset T$ and
sets $F_{1}\in \mathcal{F}_{\Lambda _{t_{i_{1}}}\times \cdots \times \Lambda
_{t_{i_{k}}}}$ and $F_{2}\in \mathcal{F}_{\Lambda _{t_{j_{1}}}\times \cdots
\times \Lambda _{t_{j_{m}}}}.$ Since $A_{1}\cap A_{2}$ $=\varnothing ,$ the
sets $F_{1},$ $F_{2}$ in (\ref{14_}) are such that, for a point $(\lambda
_{1},...,\lambda _{n})$ in $\Lambda _{t_{1}}\times \cdots \times \Lambda
_{t_{n}},$ conditions $(\lambda _{i_{1}},...,\lambda _{i_{k}})\in F_{1}$ and 
$(\lambda _{j_{1}},...,\lambda _{j_{m}})\in F_{2}$ are mutually exclusive,
that is: 
\begin{eqnarray}
&&\left\{ (\lambda _{1},...,\lambda _{n})\in \Lambda _{t_{1}}\times \cdots
\times \Lambda _{t_{n}}\mid (\lambda _{i_{1}},...,\lambda _{i_{k}})\in
F_{1}\right\}   \label{15_} \\
&&\cap \left\{ (\lambda _{1},...,\lambda _{n})\in \Lambda _{t_{1}}\times
\cdots \times \Lambda _{t_{n}}\mid (\lambda _{j_{1}},...,\lambda
_{j_{m}})\in F_{2}\right\}   \notag \\
&=&\varnothing .  \notag
\end{eqnarray}%
Taking into the account relations (\ref{123}), (\ref{14_}), (\ref{15_}), the
consistency conditions (\ref{5_}), (\ref{6__}) and also that each $\mathfrak{%
M}_{(t_{1},...,t_{n})}$ is a finitely additive measure, we have%
\begin{eqnarray}
&&\mathbb{M(}A_{1}\cup A_{2})  \label{16_} \\
&=&\mathfrak{M}_{(t_{1},...,t_{n})}(\{(\lambda _{1},...,\lambda _{n})\in
\Lambda _{t_{1}}\times \cdots \times \Lambda _{t_{n}}\mid (\lambda
_{i_{1}},...,\lambda _{i_{k}})\in F_{1}\text{ \ or \ }(\lambda
_{j_{1}},...,\lambda _{j_{m}})\in F_{2}\})  \notag \\
&=&\mathfrak{M}_{(t_{1},...,t_{n})}\left( \left\{ (\lambda _{1},...,\lambda
_{n})\in \Lambda _{t_{1}}\times \cdots \times \Lambda _{t_{n}}\mid (\lambda
_{i_{1}},...,\lambda _{i_{k}})\in F_{1}\right\} \right)   \notag \\
&&+\mathfrak{M}_{(t_{1},...,t_{n})}\left( \left\{ (\lambda _{1},...,\lambda
_{n})\in \Lambda _{t_{1}}\times \cdots \times \Lambda _{t_{n}}\mid (\lambda
_{j_{1}},...,\lambda _{j_{m}})\in F_{2}\right\} \right)   \notag \\
&=&\mathfrak{M}_{(t_{i_{1}},...,t_{i_{k}})}\left( F_{1}\right) +\mathfrak{M}%
_{(t_{j_{1}},...,t_{j_{m}})}\left( F_{2}\right)   \notag \\
&=&\mathbb{M(}A_{1})+\mathbb{M(}A_{2}).  \notag
\end{eqnarray}%
Therefore, the normalized set function $\mathbb{M}$ on $\mathcal{A}_{%
\widetilde{\Lambda }}$ defined via relation (\ref{123}) is additive and
constitutes a finitely additive measure on $\mathcal{A}_{\widetilde{\Lambda }%
}$, see remark 1.

Thus, the set function $\mathbb{M}:\mathcal{A}_{\widetilde{\Lambda }%
}\rightarrow \mathcal{L}_{\mathcal{H}}^{(s)}$ defined by relation (\ref{123}%
) constitutes a unique normalized finitely additive $\mathcal{L}_{\mathcal{H}%
}^{(s)}$-valued measure satisfying representation (\ref{8_}). This completes
the proof of lemma 3.
\end{proof}

\section{proof of lemma 2}

For a quantum observable $X\in \mathfrak{X}_{\mathcal{H}}$, bounded or
unbounded, and a Borel function $\varphi :\mathbb{R\rightarrow R}$, the
quantum observable $\varphi (X)\equiv \varphi \circ X$ is defined as \cite%
{reed}%
\begin{equation}
\varphi (X):=\dint \limits_{\mathbb{R}}\varphi (x)\mathrm{P}_{X}(\mathrm{d}x)
\end{equation}%
and has the spectrum $\mathrm{sp}\varphi (X)=\varphi (\mathrm{sp}X).$ Its
spectral measure $\mathrm{P}_{\varphi (X)}\ $satisfies the relation%
\begin{equation}
\mathrm{P}_{\varphi (X)}(B)=\mathrm{P}_{X}(\varphi ^{-1}(B)),\text{ \ \ }%
B\in \mathcal{B}_{\mathrm{sp}\varphi (X)}.  \label{wd}
\end{equation}%
Due to (\ref{wd}), if quantum observables $X,Y_{1},...,Y_{m}$ mutually
commute, then the same is true for quantum observables $\varphi
(X),Y_{1},...,Y_{m}$.

For arbitrary mutually commuting observables $X,Y_{1},...,Y_{m}$ on $%
\mathcal{H},$ relation (\ref{28''}) and property (\ref{wd}) imply 
\begin{eqnarray}
&&\nu _{\rho }\mathbb{(}f_{\varphi (X)}^{-1}(B)\cap
f_{Y_{1}}^{-1}(B_{1})\cap \cdots \cap f_{Y_{m}}^{-1}(B_{m}))  \label{z} \\
&=&\mathrm{tr}[\rho \{ \mathrm{P}_{\varphi (X)}(B)\cdot \mathrm{P}%
_{Y_{1}}(B_{1})\cdot \ldots \cdot \mathrm{P}_{Y_{m}}(B_{m})\}]  \notag \\
&=&\mathrm{tr}[\rho \left \{ \mathrm{P}_{X}(\varphi ^{-1}(B))\cdot \mathrm{P}%
_{Y_{1}}(B_{1})\cdot \ldots \cdot \mathrm{P}_{Y_{m}}(B_{m})\right \} ] 
\notag \\
&=&\nu _{\rho }(f_{X}^{-1}(\varphi ^{-1}(B))\cap f_{Y_{1}}^{-1}(B_{1})\cap
\cdots \cap f_{Y_{m}}^{-1}(B_{m}))  \notag \\
&=&\nu _{\rho }((\varphi \circ f_{X})^{-1}(B)\cap f_{Y_{1}}^{-1}(B_{1})\cap
\cdots \cap f_{Y_{m}}^{-1}(B_{m})).  \notag
\end{eqnarray}%
This proves lemma 2.

\section{proof of theorem 3 and corollary 3}

In order to prove the existence point of theorem 3, let us take the specific
statistically noncontextual qHV model that we used for the proof of theorem
2. Namely, let $(\Lambda ,\mathcal{F}_{\Lambda })$ and $\pi _{X}:\Lambda
\rightarrow \mathrm{sp}X,$ $X\in \mathfrak{X}_{\mathcal{H}}$, be the
measurable space and the random variables specified in theorem 1. Then the
one-to-one mapping $\widetilde{\Phi }:\mathfrak{X}_{\mathcal{H}}\rightarrow 
\mathfrak{F}_{(\Lambda ,\mathcal{F}_{\Lambda })},$ defined by $\ \widetilde{%
\Phi }(X)=\pi _{X},$ $\forall X\in \mathfrak{X}_{\mathcal{H}},$ satisfies
the setting of theorem 2.

Note that, for each Borel function $\varphi :\mathbb{R\rightarrow R}$ and
each observable $X,$ the random variable $\varphi \circ \pi _{X}:\Lambda
\rightarrow \mathrm{sp}\varphi (X)$ satisfies the spectral correspondence
rule $(\varphi \circ \pi _{X})(\Lambda )=\mathrm{sp}\varphi (X)$ but does
not coincide with the random variable $\pi _{\varphi (X)}$ and does not
belong to the image $\widetilde{\Phi }(\mathfrak{X}_{\mathcal{H}})\mathfrak{.%
}$ However, due to representation (\ref{28''}) and property (\ref{pp}), the
random variables $\varphi \circ \pi _{X}$ and $\pi _{\varphi (X)}$
equivalently model the quantum observable $\varphi (X)$ under all joint von
Neumann measurements, \emph{independently} of their measurements contexts.

Taking all this into the account, let $\phi _{X}^{(\theta )}(Y_{\theta })=X,$
with a Borel function $\phi _{X}^{(\theta )}:\mathbb{R}\rightarrow \mathbb{R}%
,$ an observable $Y_{\theta }\in \mathfrak{X}_{\mathcal{H}}$ and $\theta \in
\Theta _{X},$ be a possible functional representation of a quantum
observable $X\in \mathfrak{X}_{\mathcal{H}}$ via a quantum observable $%
Y_{\theta }$. Here, $\Theta _{X}$ is an index set of all such
representations of an observable $X.$ Since there is always the trivial
representation where $Y_{\theta _{0}}\equiv X,$ $\phi _{X}^{(\theta
_{0})}(X)\equiv X,$ the set $\Theta _{X}$ is non-empty for each $X.$

Denote by $\mathfrak{[}\pi _{X}]$ the set of the following random variables
on $(\Lambda ,\mathcal{F}_{\Lambda }):$%
\begin{eqnarray}
\mathfrak{[}\pi _{X}] &=&\{g_{X}^{(\theta )}\mid \theta \in \Theta _{X}\},%
\text{\ \ \ \ where \ }g_{X}^{(\theta _{0})}=\pi _{X},\text{ \ \ }%
g_{X}^{(\theta )}=\phi _{X}^{(\theta )}\circ \pi _{Y_{\theta }},  \label{ss'}
\\
\phi _{X}^{(\theta )}\circ Y_{\theta } &=&X,\text{ \ \ }\phi _{X}^{(\theta
)}:\mathbb{R}\rightarrow \mathbb{R},\text{ \ \ }Y_{\theta }\in \mathfrak{X}_{%
\mathcal{H}}.  \notag
\end{eqnarray}%
By lemma 2, all random variables in $\mathfrak{[}\pi _{X}]$ equivalently
represent an observable $X$ under all joint von Neumann measurements.
Moreover, by our construction of the set $\mathfrak{[}\pi _{X}],$ each
random variable $g_{X}^{(\theta )}\in $ $\mathfrak{[}\pi _{X}]$ satisfies
the spectral correspondence rule $g_{X}^{(\theta )}(\Lambda )=\mathrm{sp}X.$

In order to prove that 
\begin{equation}
\ X_{1}\neq X_{2}\text{ \ \ }\Rightarrow \text{ \ \ }\mathfrak{[}\pi
_{X_{1}}]\cap \mathfrak{[}\pi _{X_{2}}]=\varnothing ,  \label{ll}
\end{equation}%
let us suppose that, for some observables $X_{1}\neq X_{2},$ the
intersection $\mathfrak{[}\pi _{X_{1}}]\cap \mathfrak{[}\pi _{X_{2}}]\neq
\varnothing .$ Then $\mathrm{sp}X_{1}=\mathrm{sp}X_{2}$ and a common random
variable $g\in \mathfrak{[}\pi _{X_{1}}]\cap \mathfrak{[}\pi _{X_{2}}]$
admits either of representations 
\begin{eqnarray}
g &=&\phi _{X_{1}}\circ \pi _{Y_{1}},\text{ \ \ }g=\phi _{X_{2}}\circ \pi
_{Y_{2}},\text{ \ \ where}  \label{nn'} \\
\phi _{X_{1}}(Y_{1}) &=&X_{1},\text{\ \ \ }\phi _{X_{2}}(Y_{2})=X_{2},\text{
\ \ }Y_{1},Y_{2}\in \mathfrak{X}_{\mathcal{H}},  \notag
\end{eqnarray}%
with 
\begin{eqnarray}
\pi _{Y_{1}}^{-1}(\phi _{_{X_{1}}}^{-1}(B)) &=&(\phi _{X_{1}}\circ \pi
_{Y_{1}})^{-1}(B)  \label{cb} \\
&=&(\phi _{X_{2}}\circ \pi _{Y_{2}})^{-1}(B)=\pi _{Y_{2}}^{-1}(\phi
_{_{X_{2}}}^{-1}(B))  \notag
\end{eqnarray}%
for each $B\in \mathcal{B}_{\mathrm{sp}X_{1}}=\mathcal{B}_{\mathrm{sp}%
X_{2}}. $ Due to (\ref{wd}), (\ref{21}) and (\ref{cb}), we have:%
\begin{eqnarray}
\mathrm{P}_{X_{1}}(B) &=&\mathrm{P}_{\phi _{_{X_{1}}}(Y_{1})}(B)=\mathrm{P}%
_{Y_{1}}(\phi _{_{X_{1}}}^{-1}(B))  \label{eq} \\
&=&\mathbb{M(}\pi _{Y_{1}}^{-1}(\phi _{_{X_{1}}}^{-1}(B))  \notag \\
&=&\mathbb{M((}\pi _{Y_{2}}^{-1}(\phi _{_{X_{2}}}^{-1}(B)))  \notag \\
&=&\mathrm{P}_{Y_{2}}(\phi _{_{X_{2}}}^{-1}(B))=\mathrm{P}_{\phi
_{_{X_{2}}}(Y_{2})}(B)=\mathrm{P}_{X_{2}}(B),\text{ \ \ \ }  \notag \\
B &\in &\mathcal{B}_{\mathrm{sp}X_{1}}=\mathcal{B}_{\mathrm{sp}X_{2}}. 
\notag
\end{eqnarray}%
In view of the spectral theorem (\ref{2}), relation $\mathrm{P}_{X_{1}}(B)=%
\mathrm{P}_{X_{2}}(B),$ $\forall B,$ implies $X_{1}=X_{2}$. Thus, if $%
\mathfrak{[}\pi _{X_{1}}]\cap \mathfrak{[}\pi _{X_{2}}]\neq \varnothing $,
then $X_{1}=X_{2}.$ This proves (\ref{ll}).

Let 
\begin{equation}
\mathfrak{F}_{\mathfrak{X}_{\mathcal{H}}}:=\dbigcup \limits_{X\in \mathfrak{X%
}_{\mathcal{H}}}\mathfrak{[}\pi _{X}]  \label{hj}
\end{equation}%
be the union of all disjoint sets $\mathfrak{[}\pi _{X}],$ $\forall X\in 
\mathfrak{X}_{\mathcal{H}},$ of random variables and $\Psi :\mathfrak{F}_{%
\mathfrak{X}_{\mathcal{H}}}\rightarrow \mathfrak{X}_{\mathcal{H}}$ be the
mapping defined via the relation%
\begin{equation}
\Psi (g)=X,\text{ \ \ }g\in \mathfrak{[}\pi _{X}]\subset \mathfrak{F}_{%
\mathfrak{X}_{\mathcal{H}}},\text{ \ }X\in \mathfrak{X}_{\mathcal{H}},
\label{def}
\end{equation}%
and having, therefore, the preimage 
\begin{equation}
\Psi ^{-1}(\{X\})=\mathfrak{[}\pi _{X}],\text{ \ \ }X\in \mathfrak{X}_{%
\mathcal{H}}.
\end{equation}

For a Borel function $\varphi :\mathbb{R\rightarrow R}$ and an arbitrary $%
g_{X}^{(\theta )}\in \mathfrak{[}\pi _{X}]=\Psi ^{-1}(\{X\}),$ consider the
random variable $\varphi \circ g_{X}^{(\theta )}$. Due to our above
construction (\ref{ss'}) of the set $\mathfrak{[}\pi _{X}],$ a random
variable $g_{X}^{(\theta )}$ admits a representation 
\begin{equation}
g_{X}^{(\theta )}=\phi _{X}^{(\theta )}\circ \pi _{Y_{\theta }},
\end{equation}%
where a Borel function $\phi _{X}^{(\theta )}:\mathbb{R}\rightarrow \mathbb{R%
}$ and a quantum observable $Y_{\theta }$ are such that $\phi _{X}^{(\theta
)}\circ Y_{\theta }=X.$ We have: 
\begin{eqnarray}
\varphi \circ g_{X}^{(\theta )} &=&(\varphi \circ \phi _{X}^{(\theta
)})\circ \pi _{Y_{\theta }} \\
&\in &\mathfrak{[}\pi _{(\varphi \circ \phi _{_{X}}^{(\theta )})\circ
Y_{\theta }}]=\mathfrak{[}\pi _{\varphi \circ X}]=\Psi ^{-1}(\{ \varphi
\circ X\}).  \notag
\end{eqnarray}

Thus, the mapping $\Psi ,$ defined by relation (\ref{def}), satisfies the
functional condition (\ref{ff}).

Combining all this with representation (\ref{28''}), property (\ref{pp}),
definitions (\ref{ss'}), (\ref{def}) we prove the existence point of theorem
3, in particular, the context-independence of representation (\ref{kl}).

The relation $\sum \alpha _{j}\rho _{j}\overset{\mathfrak{R}}{\mapsto }\sum
\alpha _{j}\nu _{\rho _{j}}$ in theorem 3 follows explicitly from relation (%
\ref{ma}) in proposition 1. This completes the proof of theorem 3.

In corollary 3, the context-independent representations (\ref{gg'}), (\ref%
{gg}) follow from (\ref{kl}), (\ref{qww}), (\ref{2}) and are proved quite
similarly to our proof of (\ref{2d}), (\ref{29}). 


\begin{thebibliography}{99}
\bibitem{4-} In this pair, $\Omega $ is a set and $\mathcal{F}_{\Omega }$ is
an algebra of subsets of $\Omega .$
\bibitem{12-} A correlation scenario admits a general LqHV model iff it
admits a deterministic LqHV model, see Ref. 12 for details.
\bibitem{16-} In theorem 2, the set $f^{-1}(F):=\{ \omega \in \Omega \mid $ $%
f(\omega )\in F$ $\} \in \mathcal{F}_{\Omega }$ is the preimage of a set $F$
under the mapping $f$ and $\chi _{B}(\cdot )$ is the indicator function of a
set $B,$ that is, $\chi _{B}(\omega )=1$ for $\omega \in B$ and $\chi
_{B}(\omega )=0$ for $\omega \notin B.$
\end{thebibliography}
\end{document}